\documentclass[11pt]{article}
\usepackage{amssymb}
\usepackage{latexsym,amsmath,amsthm,amsfonts,graphicx,float}
\usepackage{subfig}
\usepackage{authblk}
\usepackage{color}
\usepackage[authoryear]{natbib}
\usepackage{asymptote}

\newtheorem{lem}{Lemma}
\newtheorem{thm}{Theorem}

\newcommand{{\convd}}{{\ \buildrel d \over \longrightarrow \ }}

\newcommand{\be}{\begin{equation}}
\newcommand{\ee}{\end{equation}}

\begin{document}

\bibliographystyle{apalike}


\title{Direction-Projection-Permutation for High Dimensional Hypothesis Tests}
\author[a,1]{Susan Wei}
\author[a,2]{Chihoon Lee}
\author[b,3]{Lindsay Wichers}
\author[a]{Gen Li}
\author[a]{J.S. Marron}

\affil[a]{Department of Statistics and Operations Research, University of North Carolina - Chapel Hill.}
\affil[b]{Department of Environmental Sciences and Engineering, School of Public Health, University of North Carolina - Chapel Hill.}

\maketitle

\footnotetext[1]{To whom correspondence should be addressed. Address: Department of Statistics and Operations Research, 318 Hanes Hall, CB\# 3260, UNC Chapel Hill, NC 27599-3260, USA. Phone: 626-375-0351. Email: susanwe@live.unc.edu}
\footnotetext[2]{Current Affilation: Department of Statistics, 221 Statistics Building, Colorado State University, Fort Collins, CO 80523-1877, USA. }
\footnotetext[3]{Current Affiliation: Environmental Media Assessment Group, MD B243-01, National Center for Environmental Assessment, Office of Research and Development, U.S. Environmental Protection Agency, Research Triangle Park, North Carolina 27711, USA.}

\begin{abstract}

Motivated by the prevalence of high dimensional low sample size datasets in modern statistical applications, we propose a general nonparametric framework, Direction-Projection-Permutation (DiProPerm), for testing high dimensional hypotheses. 
The method is aimed at rigorous testing of whether lower dimensional visual differences are statistically significant. 
Theoretical analysis under the non-classical asymptotic regime of dimension going to infinity for fixed sample size reveals that certain natural variations of DiProPerm can have very different behaviors. An empirical power study both confirms the theoretical results and suggests DiProPerm is a powerful test in many settings. Finally DiProPerm is applied to a high dimensional gene expression dataset. 

\end{abstract}


\section{Introduction}
\label{Introduction}

We propose a nonparametric procedure for testing high dimensional hypotheses that is especially practical in high dimensional low sample size (HDLSS) settings. HDLSS data sets arise in many modern applications of statistics, including genetics, chemometrics, and image analysis. 
An intuitive approach to looking for differences between two high dimensional distributions is by looking for differences between their one dimensional projections onto some appropriate direction. DiProPerm is a three-stepped procedure based on this idea.  The procedure is as follows:
\begin{enumerate}
\item[1.] Direction --- take the normal vector to the separating hyperplane of a binary linear classifier trained on the class labels.
\item[2.] Projection --- project data from both samples onto this direction and calculate a univariate two-sample statistic. An illustration of this can be seen in the first panel of Figure \ref{introFig}.
\item[3.] Permutation --- assess the significance of this univariate statistic by a permutation test. Namely, (a) pool the two samples and permute the class labels; (b) take the normal vector to the  binary linear classifier retrained on the permuted class labels; (c) project data onto this direction and re-calculate the univariate two-sample statistic. An illustration of this can be seen in the last three panels of Figure \ref{introFig}. For a level $\alpha$ test, we reject the null if the original test statistic is among the $100\alpha\%$ largest of the permuted statistics.
\end{enumerate}
DiProPerm is not a single test but a general hypothesis testing framework. The number of combinations of direction and univariate statistic is large. We will focus on a select few in this paper but more options are discussed in detail in Section \ref{impl} of the Supplement.

In general, we are interested in testing the hypotheses: 1) equality of two distributions and 2) equality of means. That any DiProPerm test is an exact level $\alpha$ test for equality of distributions follows immediately from general permutation test theory. A perhaps surprising point is that for testing equality of means, validity does not hold for some natural versions of DiProPerm. In this paper we study the theoretical properties of two particular DiProPerm tests. We will show that one is valid for testing equality of means while the other is not. 

\subsection{A Motivating Example}
\label{ACE}

Lower dimensional projections in directions of interest are often used to understand structure in high dimensional data. One example is the directions found by applying Principal Component Analysis (PCA), see \cite{Jolliffe2002} for an excellent introduction, which yields directions maximizing variation. When there are two classes however, as in the case we are studying, additional insights come from directions based on binary linear classifiers, where a binary classification decision is based on the value of a linear combination of the data features. 

In very high dimensions many linear classifiers over-fit. Here is a simple example illustrating this. Draw two independent samples, each of size 50, from the 1000-variate standard Gaussian distribution. We use the Distance Weighted Discrimination (DWD) direction in step 1 of DiProPerm \citep{Marron2007}. DWD is a binary linear classifier similar to the Support Vector Machine (SVM) with certain advantages in high dimensions, see  \cite{Cortes1995} for an introduction to SVM. 

The first panel of Figure \ref{introFig} shows the one dimensional projection of the data onto the DWD direction trained on the original class labels. Colors are used to represent original class membership and are thus constant throughout the first three panels. The projections are jittered on the y-axis to allow easy visualization. A kernel density estimate of the projections is plotted in the background (solid black line). We see that the projections in the first panel of Figure \ref{introFig} are very well separated despite the fact that the samples arise from the same underlying distribution. This clear over-fitting artifact common in HDLSS data is a strong motivation for DiProPerm.

\begin{figure}[h!]
\centering
\includegraphics[width=.9\textwidth]{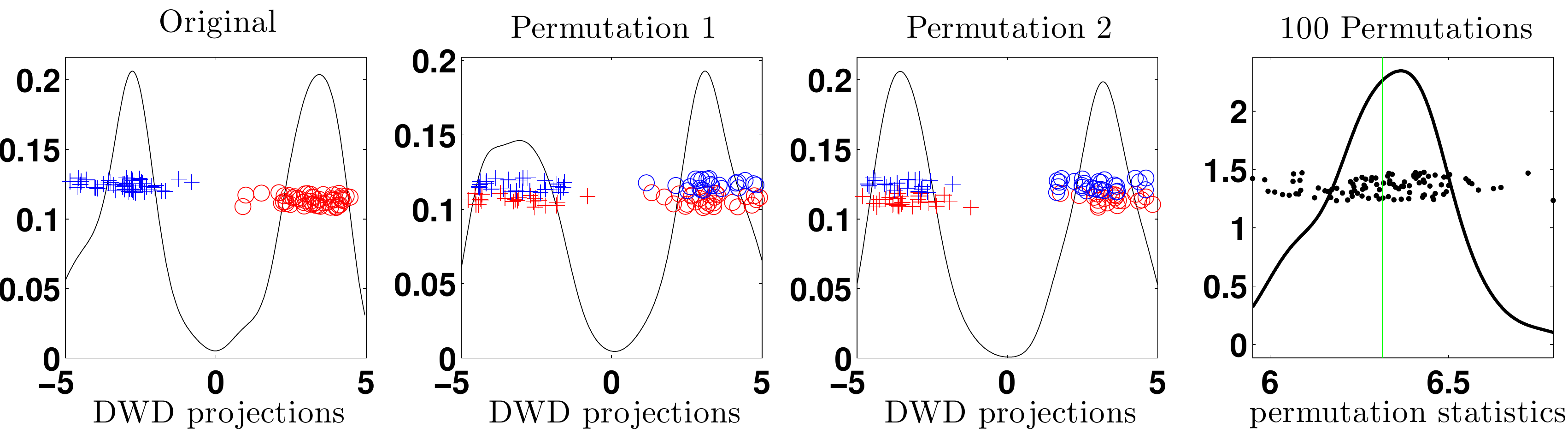}
\caption{The data are standard 1000-variate Guassian. In the first panel, the DWD direction is trained on the original class labels, represented by colors (same in all panels). In the second and third panels, the DWD directions are trained on realizations of randomly permuted class labels, represented by symbols (different in each panel). The separation in the first panel is comparable to that in the second and third panels. One hundred permutation statistics resulting from a DiProPerm test are shown in the last panel which confirms the separation in the first panel is not significant.} 
\label{introFig}
\label{introFigAnswer}
\end{figure}

The middle two panels of Figure \ref{introFigAnswer} show projections of the data onto re-trained DWD directions, each based on a realization of randomly permuted class labels. Symbols are used to represent permuted class labels and are thus different in the first three panels.
We find the projections here to be well separated with respect to the symbols. Relative to the second and third panels, the original separation we observed in the first panel is quite unremarkable, suggesting that the two underlying distribution are not different. 

The last panel in Figure \ref{introFig} confirms this observation. We perform a DiProPerm test with 100 permutations and display the statistic, chosen here to be the difference of sample means, calculated for each permutation. The vertical line is the original statistic calculated on the unprojected data. We see that based on the DiProPerm test, the null hypothesis of equal distributions should not be rejected. 

\subsection{The Hypotheses}
Let $X_1,\ldots,X_{m}$ and $Y_1,\ldots, Y_{n}$ be independent random samples of $\mathbb{R}^d$-valued random vectors, $d \ge 1$ with distributions $F_1$ and $F_2$, respectively. We are interested in testing the null hypothesis of equality of distributions
\begin{equation}
\label{strongNull}
H_0: F_1=F_2 \quad \text{ versus } \quad 
H_1: F_1\ne F_2
\end{equation}
Let $\mu(F)$ denote the mean of a distribution $F$. Another item of interest is to test the weaker null hypothesis of equality of means
\begin{equation}
\label{weakNull}
H_0: \mu(F_1) = \mu(F_2) \quad \text{ versus } \quad 
H_1: \mu(F_1) \ne \mu(F_2)
\end{equation}
Note that the multivariate Behrens-Fisher problem concerns testing \eqref{weakNull} under normality. 

\subsection{Overview}
The outline for the paper is as follows. A review of related work is presented in Section \ref{relwork}. In Section \ref{alg}, two DiProPerm tests are closely examined. HDLSS asymptotics are used to investigate the validity of these two tests for the weaker null hypothesis of equality of means in Section \ref{validity}. In Section \ref{power} we perform a Monte Carlo power study comparing DiProPerm to other methods. Finally in Section \ref{App}, DiProPerm is applied to a real microarray dataset.

\section{Related work}
\label{relwork}
There is extensive literature on testing equality of distributions for two multivariate distributions under the classical setting of sample size larger than dimension. For the more challenging HDLSS setting, several methods have been developed and we discuss two of them here. 

First, there are nearest neighbor tests \cite{Bickel1983, Henze1988, Schilling1986} which are based on nearest neighbor coincidences - the number of neighbors around a data point that belong to the same sample. The null distribution of the test statistic can be derived parametrically using normal theory or nonparametrically using a permutation test. A more recent contribution to testing equality of distributions under HDLSS settings is Szekely and Rizzo's nonparametric energy test \citep{Szekely2004}. The energy test statistic is based on the Euclidean distance between pairs of sample elements. Here significance is accessed through permutation testing. 

The nearest neighbor test and the energy test require calculation of all pairwise distances between sample elements. The computational complexity of both tests is independent of dimension, and is thus suitable for the HDLSS setting. In Section \ref{empStudyDist} we perform an empirical power study comparing DiProPerm to the energy test. 

For testing equality of means for two multivariate distributions, the classical Hotelling $T^2$ test is often used in the setting of sample size larger than dimension. However, the Hotelling $T^2$ statistic is not computable in HDLSS situations because the covariance matrix is not of full rank. To address this issue, the methods in \citep{Bai1996},\citep{Chen2010}, and \citep{Srivastava2008} replace  the covariance matrix in the Hotelling $T^2$ statistic by a diagonalized version. 

Taking a different approach, the method proposed by Lopes et al. projects the high dimensional data onto a random subspace of low enough dimension so that the traditional Hotelling $T^2$ statistic may be used \citep{Lopes2011}. All of these tests have the disadvantage that equal covariances are assumed, which is not a restriction we place on DiProPerm. In Section \ref{empStudyMeans} we perform an empirical power study comparing DiProPerm to the Random Projection test proposed in \cite{Lopes2011}. 


\section{The Choice of The Univariate Statistic}
\label{alg}

Here, we study the difference between two particular choices of the univariate statistic in Step 2 of DiProPerm. First, let the Mean Difference (MD) direction be the vector connecting the centroids of each sample. For simplicity, we will use this particular direction to compare two natural statistics of the projections: 1) the Mean Difference (MD) statistic --- the difference of sample means, and 2) the two-sample t-statistic (t) --- difference of sample means divided by $\{s_1/m + s_2/n\}^{1/2}$ where $s_1$ and $s_2$ are sample standard deviations of each class, sized $m$ and $n$ respectively. Henceforth we specify different DiProPerm tests by concatenating the direction name and two sample univariate statistic name. Following this convention, the DiProPerm test that uses the MD direction and the MD statistic will be referred to as the \textbf{MD-MD} test and the DiProPerm test that uses the MD direction and the two-sample t statistic as the \textbf{MD-t} test. 

%

We provide a toy example to contrast the difference between the MD and t statistic. We draw independent samples, each of size 50, where the first sample arises from the 1000-variate standard Gaussian distribution and the second the 1000-variate distribution with iid marginal $t(5)$ distributions. Note that the samples arise from \emph{different} distributions that have the \emph{same} means. Figure \ref{algProjs} shows the one dimensional projection of the data onto various MD directions and the MD and t statistic applied to these projections. 
\begin{figure}[h]
\begin{center}
\includegraphics[width=\textwidth]{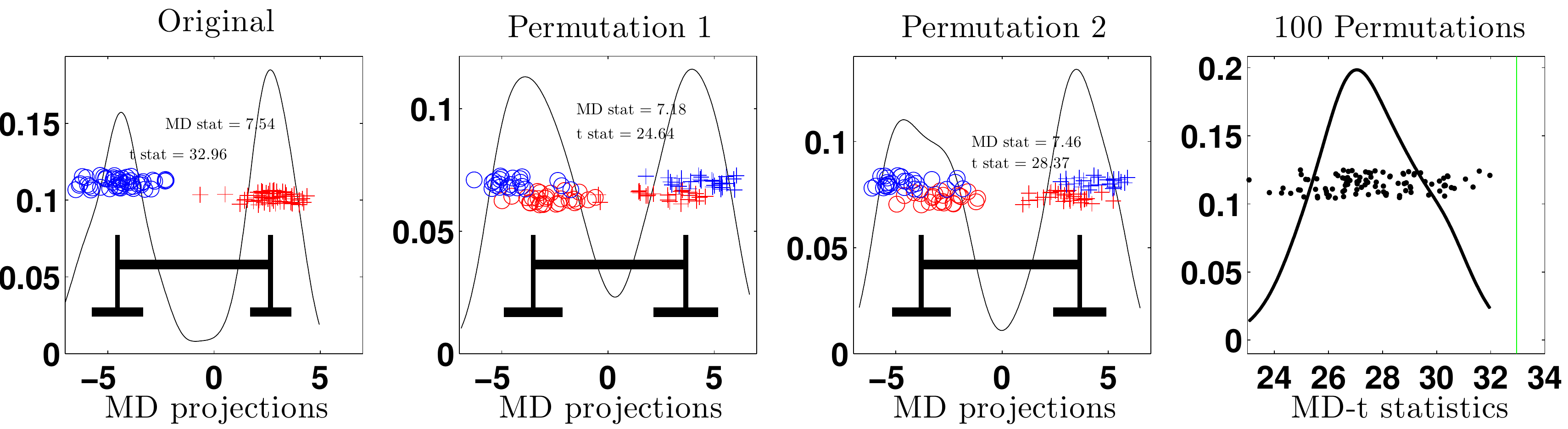}
\caption{The first sample arises from a standard 1000-variate Gaussian distribution and the second sample arises from the 1000-variate distribution with iid $t(5)$ marginals. In the first panel, the MD direction is trained on the original class labels, represented by colors. In subsequent panels, the MD direction is trained on realizations of permuted class labels, represented by symbols. Note that the MD-MD statistic is similar across the first three panels while the MD-t statistic is much larger in the first panel. One thousand permutation MD-t statistics are shown in the last panel. The empirical p-value is small suggesting the test that uses MD-t would reject the null.} 
\label{algProjs}
\end{center}
\end{figure}

The lengths of the longer horizontal black bars represent the MD statistics while the lengths of the shorter horizontal bars represent the sample standard deviations of the projected data in each permuted group. The MD statistic and two-sample t statistic calculated on the projected data are displayed towards the top of each panel. We see that the t-statistic in the first panel is much higher than the permuted t-statistics in the second and third panels. On the other hand, the MD statistic is about the same between the original and permuted worlds. We confirm this is a systematic pattern by looking at 1000 permutations and calculating the MD-t statistic. The distribution of the permuted MD-t statistics can be seen in the last panel of Figure \ref{algProjs}. We see that the original MD-t statistic, represented as a vertical line, is among the larger permutation statistics, leading us to reject the null hypothesis. The distribution of the MD-MD permutation statistics, not shown here, looks very similar to the last panel of Figure \ref{introFigAnswer}, where the original statistic is close to the middle of the permutation distribution. Thus under this setting the MD-t test rejects the null while the MD-MD does not.

This apparent inconsistency is due to the fact that the MD and t statistics are actually testing different hypotheses. The former is testing the weak hypothesis of equality of means while the latter is testing the strong hypothesis of equality of distributions. In light of this, each test is correct in its decision. This phenomenon is studied in detail in the next section.

\section{Hypothesis Test Validity}
\label{validity}
In this section, we study the validity of the MD-MD and the MD-t for testing 1) equality of distributions and 2) equality of means. We work with the MD direction because it is most amenable to theoretical analysis. Future work will include other directions such as DWD, SVM, etc. High dimensional geometric representation of SVM and DWD described in \cite{Bolivar-Cime2013} could provide the basis for this endeavor.

That both the MD-MD and the MD-t are exact tests for equality of distributions follows from standard theory on permutation tests. We will discuss how an exact level $\alpha$ test can be constructed by a permutation test. Let $N=m+n$ and write 
$
Z=(Z_1,\ldots,Z_N) = (X_1,\ldots,X_m,Y_1,\ldots,Y_n)
$
for the pooled sample. Let $\{\pi(1),\ldots,\pi(N)\}$ be a permutation of $\{1,\ldots,N\}.$  Write 
$
Z_\pi=(Z_{\pi(1)},\ldots,Z_{\pi(N)})
$ 
for the permuted  sample. Let $G_N$ denote the set of all permutations $\pi$ of $\{1,\ldots,N\}$. Then for any test statistic $V_{m,n}=V_{m,n}(Z_1,\ldots,Z_N)$, we can calculate $V_{m,n}(Z_{\pi(1)},\ldots,Z_{\pi(N)})$ for all $\pi \in G_{N}$. 
%
%
The test that rejects the null if the original statistic $V_{m,n}(Z_1,\ldots,Z_N)$ is larger than $(1-\alpha) 100 \%$ of the permuted statistics $V_{m,n}(Z_{\pi(1)},\ldots,Z_{\pi(N)})$ is an exact level $\alpha$ test. The exactness comes from the fact that the unconditional distribution and the permutation distribution of the statistic coincide under the null of equal distributions. It follows that the MD-MD test and the MD-t test, and any other DiProPerm test, are exact for testing equality of distributions. 

The matter of establishing validity for testing equality of means is not as straightforward on the other hand. In general, permutation tests cannot be expected to be valid for testing weaker hypotheses such as equality of means. For instance, if the covariances are \emph{not} the same, we have to be very careful with our choice of direction and two-sample statistic. The signal in the covariances may confound our interpretation of tests that are sensitive to both the signal in the mean and the signal in the variances. This is consistent with our results which show that under normality and balanced sample sizes, the MD-MD remains valid for testing equality of means under heterogeneous covariances. On the other hand, the MD-t is invalid when the covariances are not the same.

\subsection{MD-MD}
\label{valMDMD}

In this section, we establish that the MD-MD test is an exact test for equality of means under normality and balanced sample sizes. The MD-MD test statistic, $T_{m,n}(Z)$, is the mean of the projections of the $X$'s onto the unit vector in the direction of $\bar X - \bar Y$ minus the mean of the projections of the $Y$'s onto the unit vector in the direction of $\bar X - \bar Y$:
\begin{align}
T_{m,n}(Z) &= T_{m,n}(X_1,\ldots,X_m,Y_1,\ldots,Y_n) \\
&=\frac{1}{m} \sum_{i=1}^m X_i' \frac{(\bar X - \bar Y)}{||\bar X - \bar Y||} - \frac{1}{n} \sum_{j=1}^n Y_j' \frac{(\bar X - \bar Y)}{||\bar X - \bar Y||} \\
&= || \bar X - \bar Y ||
\label{MD MD}
\end{align}

\begin{thm}
Let $X_1, \ldots, X_m$ be an iid sample from the d-variate Gaussian distribution $N(\mu_X, \Sigma_x)$ and $Y_1,\ldots,Y_n$ be an independent sample drawn iid from the d-variate Gaussian distribution $N(\mu_Y,\Sigma_y)$ where $\Sigma_X \ne \Sigma_Y$. If $m=n$ then the unconditional distribution and the permutation distribution of $T_{m,n}(Z)$ are equal under the null $\mu_X = \mu_Y$.
\label{MDMDthm}
\end{thm}

\begin{proof}
Under $\mu_X=\mu_Y$, $\bar X - \bar Y$ is distributed as
\begin{align}
\label{norm1}
N(0, {\Sigma_x }/{m} + {\Sigma_y }/{n}) 
\end{align}
and the permutation distribution of $\bar X - \bar Y$ is  
\begin{equation}
\label{norm2}
 \sum_{r=0}^{m} \frac{{ m \choose r} {n \choose r}}{{ N \choose m} } N\left (0, \frac{(m-r) \Sigma_x +r \Sigma_y }{m^2} +\frac{r \Sigma_x+(n-r) \Sigma_y }{n^2} \right)
\end{equation}
If $m=n$, the expressions in \eqref{norm1} and \eqref{norm2} are the same, in which case the unconditional and permutation distribution of $T_{m,n}(Z)$ are also the same. 
\end{proof}

\subsection{MD-t}
\label{valMDt}

The MD-t statistic, denoted by $U_{m,n}(Z)$, is the result of applying the unbalanced sample sizes, unequal variance two-sample t-test statistic (also known as Welch's t-test \citep{Welch1947}) to the projections onto the MD direction. Let $a \cdot b$ denote the standard dot product between two vectors in $\mathbb R^d$. The sample variances of the projected data can be expressed as
\begin{align*}
s_{\tilde X}^2 = \frac{1}{m-1} \sum_{i=1}^m [ (X_i - \bar X) \cdot (\bar X - \bar Y) ]^2 
\end{align*}
and 
\begin{align*}
s_{\tilde Y}^2 = \frac{1}{n-1} \sum_{i=1}^n [ (Y_i - \bar Y) \cdot (\bar X - \bar Y) ]^2 .
\end{align*}
Define $S_{m,n}(Z) = S_{m,n}(X_1,\ldots,X_m,Y_1,\ldots,Y_n) = { 	 {s_{\tilde{X}}^2}/{m}	 + 	{s_{\tilde{Y}}^2}/{n}  	 }$. The MD-t statistic is
\[
U_{m,n}(Z) = U_{m,n}(X_1,\ldots,X_m,Y_1,\ldots,Y_n) = {T_{m,n}(Z)^2}/{ \{	 S_{m,n}(Z) 	 \}^{1/2}	}
\]
where $T_{m,n}(Z)$ is as in Section \ref{valMDMD}. We use the term ``projected" rather loosely here since we have not normalized by $|| \bar X - \bar Y ||$. This is of no actual consequence since the two-sample t-statistic is scale invariant.

Under equal means the numerator in the MD-t statistic behaves similarly in the permutation world and the original world. However, we will see that the denominator of the MD-t statistic has very different behavior. We find that the denominator of the MD-t is larger in the permuted world, as seen in Figure \ref{algProjs}. This has the effect of making the unconditional distribution of the MD-t statistic larger than the permutation distribution. 

To gain some intuition, consider the following toy HDLSS example. Suppose we observe $X_1,X_2 \sim F_1$ and $Y_1,Y_2 \sim F_2$ where $F_1 = N(0,I_d)$ and $F_2 = N(0, \sigma^2 I_d)$, $\sigma^2 \ne 1$. The points $X_1,X_2,Y_1,Y_2$ form the vertices of a tetrahedron in three dimensional space. The two-dimensional plane generated by $Y_1, Y_2$ and $\bar X$ is shown in Figure \ref{triangle}. Distances between elements of interest are calculated using standard HDLSS asymptotics, see \cite{Hall2005} for examples of this type of calculation. All distances have an additional $O_P(1)$ term that is not shown to avoid clutter. The geometric configuration in Figure \ref{triangle} has the implication that $s_{\tilde Y}^2$ is small. To see this, note the projections of $Y_1$ and $Y_2$ onto the MD direction $\bar X-\bar Y$ is close to the projection of $\bar Y$ itself. A similar argument can be applied to show $s_{\tilde X}^2$ is small.
%
%
%
%

\begin{figure}[h!]
\centering
\includegraphics[width=.9\textwidth]{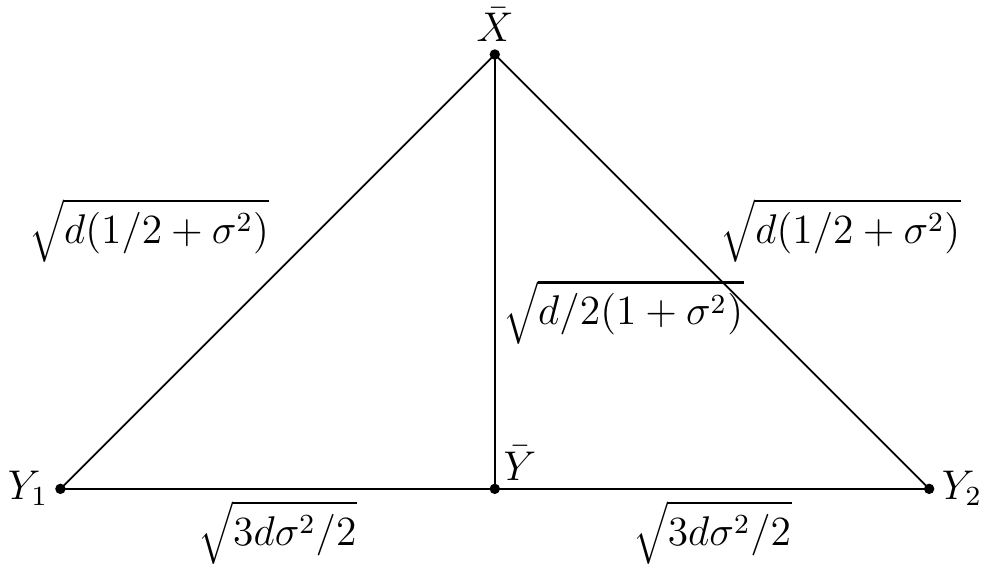}
\caption{ Plane generated by $Y_1$, $Y_2$ and $\bar X$ where $X_1,X_2 \sim F_1 = N(0,I_d)$ and  $Y_1, Y_2 \sim F_2 = N(0, \sigma^2 I_d)$ for $\sigma^2 \ne 1$. Note that the projections of $Y_1$ and $Y_2$ onto $\bar X - \bar Y$ is close to the projection of $\bar Y$ onto $\bar X - \bar Y$. This has the implication that $s_{\tilde{Y}}^2$ will be small.}
\label{triangle}
\end{figure}

Now let's look at what happens in the permutation world. Figure \ref{permTriangle} shows the two-dimensional plane generated by the realization of a random permutation where $X_1^*=X_2$, $X_2^* = Y_2$ and $Y_1^* = X_1$ and $Y_2^* = Y_1$. Notice that the distance between $Y_1^* $ and $\bar{ X^*}$ is different than the distance between $Y_2^*$ and $\bar{ X^*}$. This has the effect of making $s_{\tilde Y^*}^2$, the sample variance of the the projections of $Y_1^*$ and $Y_2^*$, large. To see this, note the projections of $Y_1^*$ and $Y_2^*$ onto the permuted MD direction are not close to the projection of $\bar Y^*$. A similar argument can be applied to show $s_{\tilde X^*}^2$, the sample variance of the projections of $X_1^*$ and $X_2^*$, is large.
The derivations for the distances shown in Figures \ref{triangle} and \ref{permTriangle} can be found in the supplement.

%
%
%
%

\begin{figure}[h!]
\centering
\includegraphics[width=.9\textwidth]{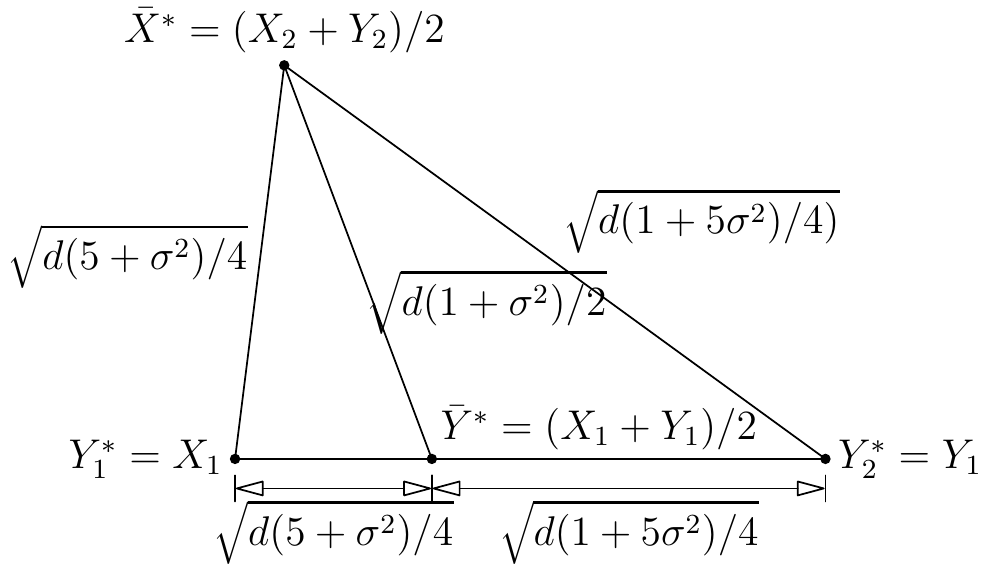}
\caption{Plane generated by a particular permutation realization of $X_1,X_2,Y_1,$ and $Y_2$. Note that the projections of $Y_1^*$ and $Y_2^*$ onto $\bar X^* - \bar Y^*$ is not close to the projection of $\bar Y^*$ onto $\bar X^* - \bar Y^*$. This has the implication that $s_{\tilde Y^*}^2$ may be large.} 
\label{permTriangle}
\end{figure}

The toy example above suggests the denominator of the MD-t statistic is larger in the permutation world than in the original world. The next result gives us a sense of just how far apart are the permutation and unconditional distributions of $S_{m,n}(Z)$.
\begin{thm}
Let $X_1, \ldots, X_m$ be a sample from the d-variate Gaussian distribution $N(\mu_x, \sigma_x^2 I_d)$ and $Y_1,\ldots,Y_n$ be an independent sample from the d-variate Gaussian distribution $N(\mu_y,\sigma_y^2 I_d)$ where $\sigma_x^2 \ne \sigma_y^2$ are scalars. Under $\mu_x = \mu_y$, we have
\[
\frac{1}{d} S_{m,n}(Z) \convd (\frac{\sigma_x^2}{m} + \frac{\sigma_y^2}{n}) \left \{\frac{1}{m-1} \frac{\sigma_x^2}{m}  \chi^2(m-1) +  \frac{1}{n-1} \frac{\sigma_y^2}{n} \chi^2(n-1) \right \}
\]
as $d$ goes to infinity.
For the permuted version, we have for some non-zero constant $c$,
\[
\frac{1}{d^2} S_{m,n}(Z_\pi) \to c \text{ in probability.}
\]
\label{MDtTheorem}
\end{thm}
The results of this theorem are surprising in that the denominator of the MD-t statistic is actually of different orders in the unconditional and permutation worlds. In particular, in the unconditional world $S_{m,n}(Z)$ grows like a random variable times $d$, while in the permutation world it grows like a constant times $d^2$.

Let us revisit the toy example earlier and see what Theorem \ref{MDtTheorem} can tell us. We make $50$ draws from $F_1 = N(0,I_d)$ and another $50$ independent draws from $F_2 = N(0,100 I_d)$. We show in Figure \ref{permVsOriginal}, using 1000 Monte Carlo realizations, the simulated permutation and unconditional distributions of the MD-t statistic for various dimensions. 

\begin{figure}[h]
\begin{center}
\includegraphics[width=\textwidth]{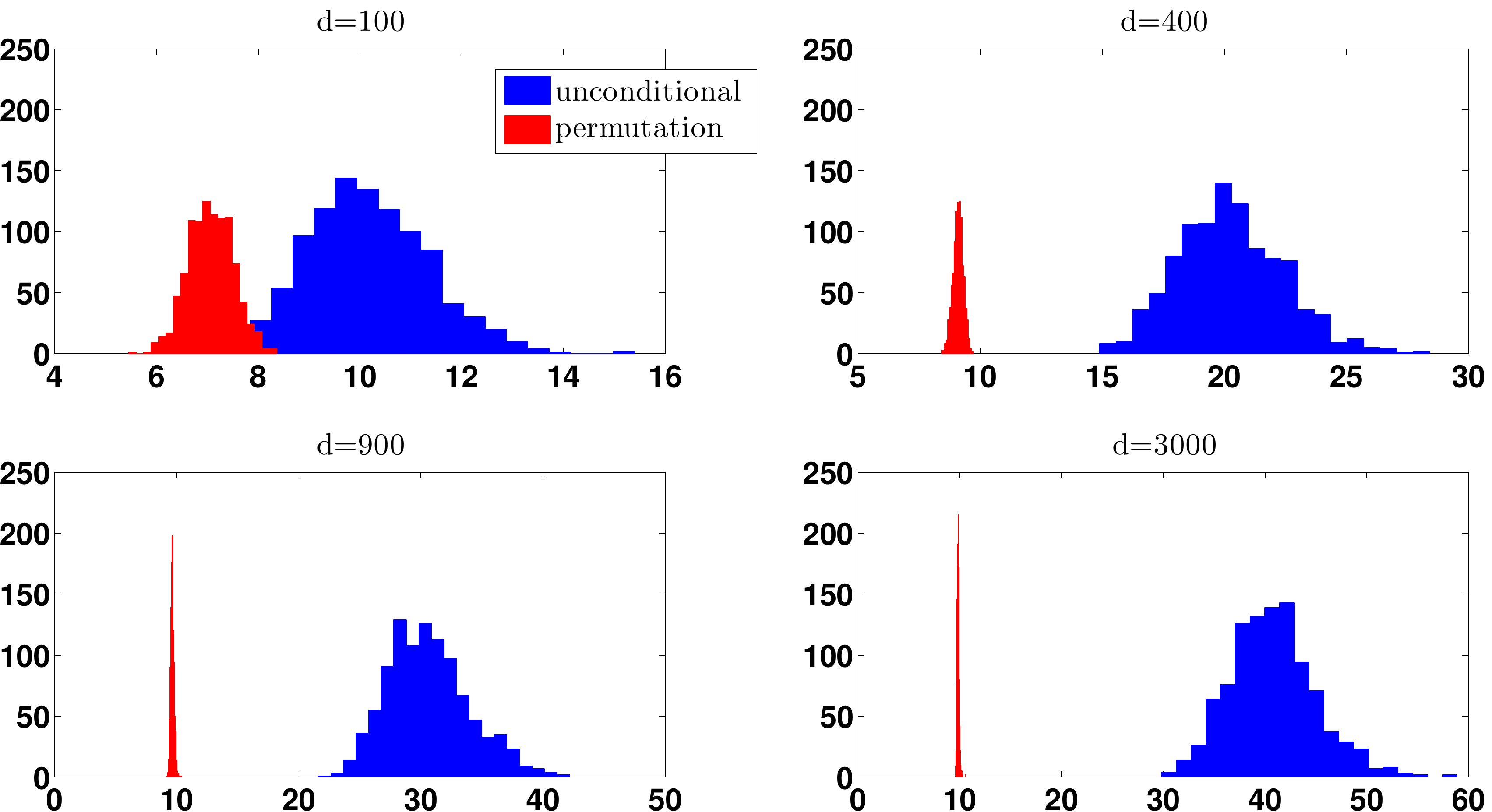}
\caption{The unconditional and permutation distribution of the MD-t statistic for the distributions $F_1 = N(0,I_d)$ and $F_2 = N(0,100 I_d)$. The separation between the unconditional and permutation distribution increases with dimension.} 
\label{permVsOriginal}
\end{center}
\end{figure}

Under the conditions in Theorem \ref{MDtTheorem}, when $\mu_x = \mu_y$, the numerator of the MD-t statistic is proportional to a $\chi^2(d)$ variable for both the unconditional and permutation distribution. On the other hand, by the results in Theorem \ref{MDtTheorem} $S_{m,n}(Z)$ is of the order $\sqrt d$ and $d$ for the unconditional and permuted distributions, respectively. Thus we should expect the MD-t statistic to be of the order $\sqrt d$ in the original unconditional world and $1$ in the permutation world. This is consistent with Figure \ref{permVsOriginal} --- the unconditional distribution is centered around $\sqrt d$ while the permutation distribution is not growing with $d$. As Figure \ref{permVsOriginal} illustrates, the unconditional distribution quickly separates from the permutation distribution as dimension increases. Thus it is very important that the MD-t statistic not be used when the goal is to test for equality of means. On the other hand, this shows the MD-t test has some power for testing equality of distributions against equal means alternatives. 

\subsection{Power surfaces}
\label{powersufs}

In this section, we study the  power of the MD-MD and MD-t for testing equality of means. In the simulations that follow, we make $m$ draws from $F_1=N(\mu_1,\sigma_1^2 I_d)$, and $n$ independent draws from $F_2=N(0,I_d)$. We set $d=500$ and $m=n=50$ for balanced sample sizes and $m=50,n=100$ for unbalanced. The dimension $d$ and sample sizes $m$ and $n$ are chosen to reflect a HDLSS setting. The significance level is set at $\alpha = 0.05$. Power is estimated using $1000$ Monte Carlo simulations.  Figure \ref{powerMdt} displays a 3D surface of power versus $\mu_1$ versus $\sigma_1^2$, using a color spectrum from cool to warm corresponding to the range 0 to 1. We also show an image underneath the surface where each pixel corresponds to the point in the 3D surface above.

Figure \ref{equalPowerMDMD} displays the estimated power surface of MD-MD under balanced sample sizes. By Theorem \ref{MDMDthm}, MD-MD is an exact test for equality of means under balanced sample sizes and normality. This is consistent with what we see in Figure \ref{equalPowerMDMD} --- when the means are equal (i.e. $\mu_1=0$), the power is around $\alpha=0.05$, as indicated by the streak at $\mu_1=0$. 
When sample sizes are unbalanced, see Figure \ref{equalPower} in Section \ref{MDsMD} of the Supplement, the MD-MD is no longer an exact test and may not even be asymptotically valid as $d \to \infty$. In Section \ref{MDsMD} of the Supplement, we propose a modification of MD-MD that should be used when sample sizes are unbalanced.

\begin{figure}[h]
\centering
\subfloat[MD-MD: Balanced]{
\includegraphics[width=.3\textwidth]{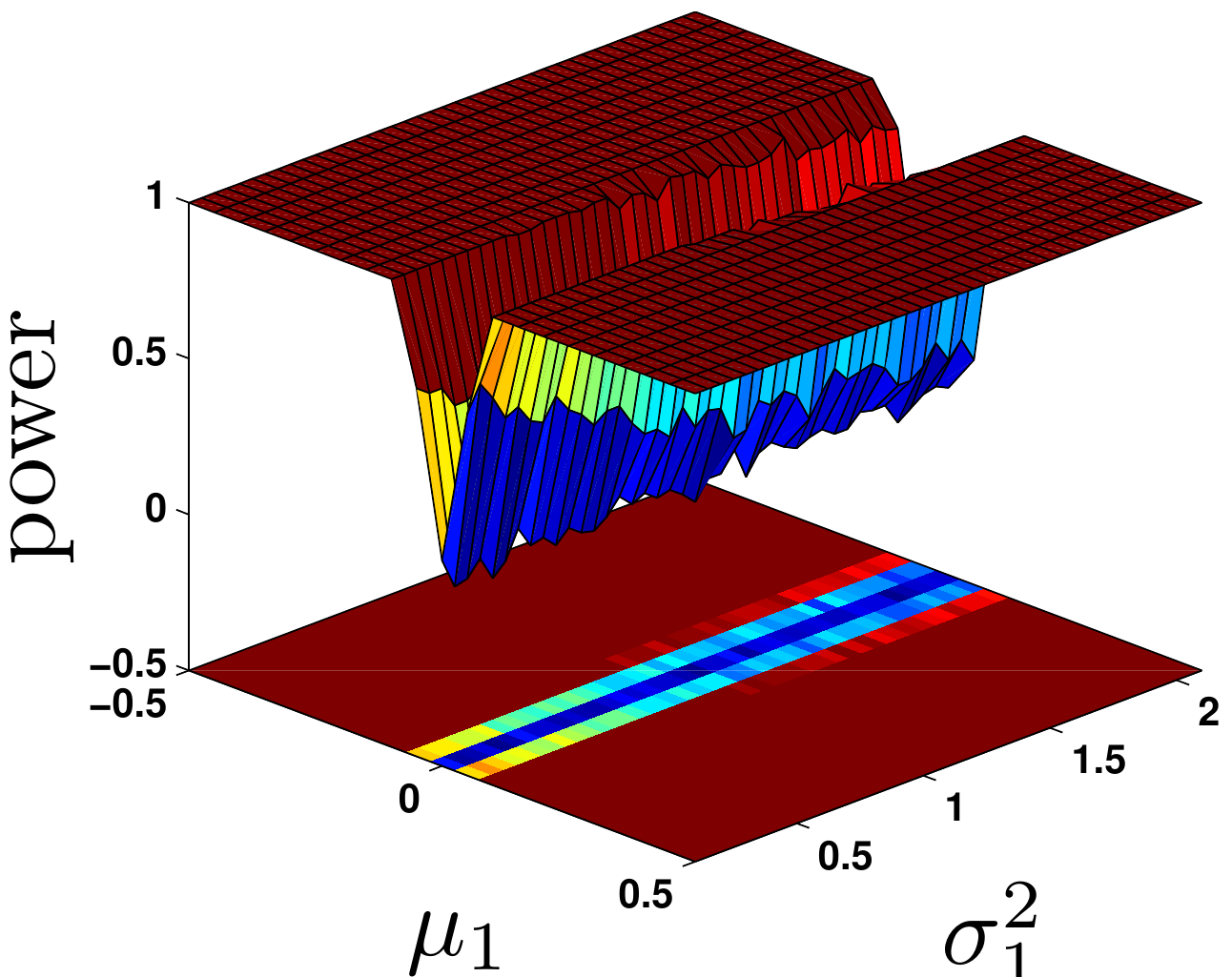}
\label{equalPowerMDMD}
}
\subfloat[MD-t: Balanced ]{
\includegraphics[width=.3\textwidth]{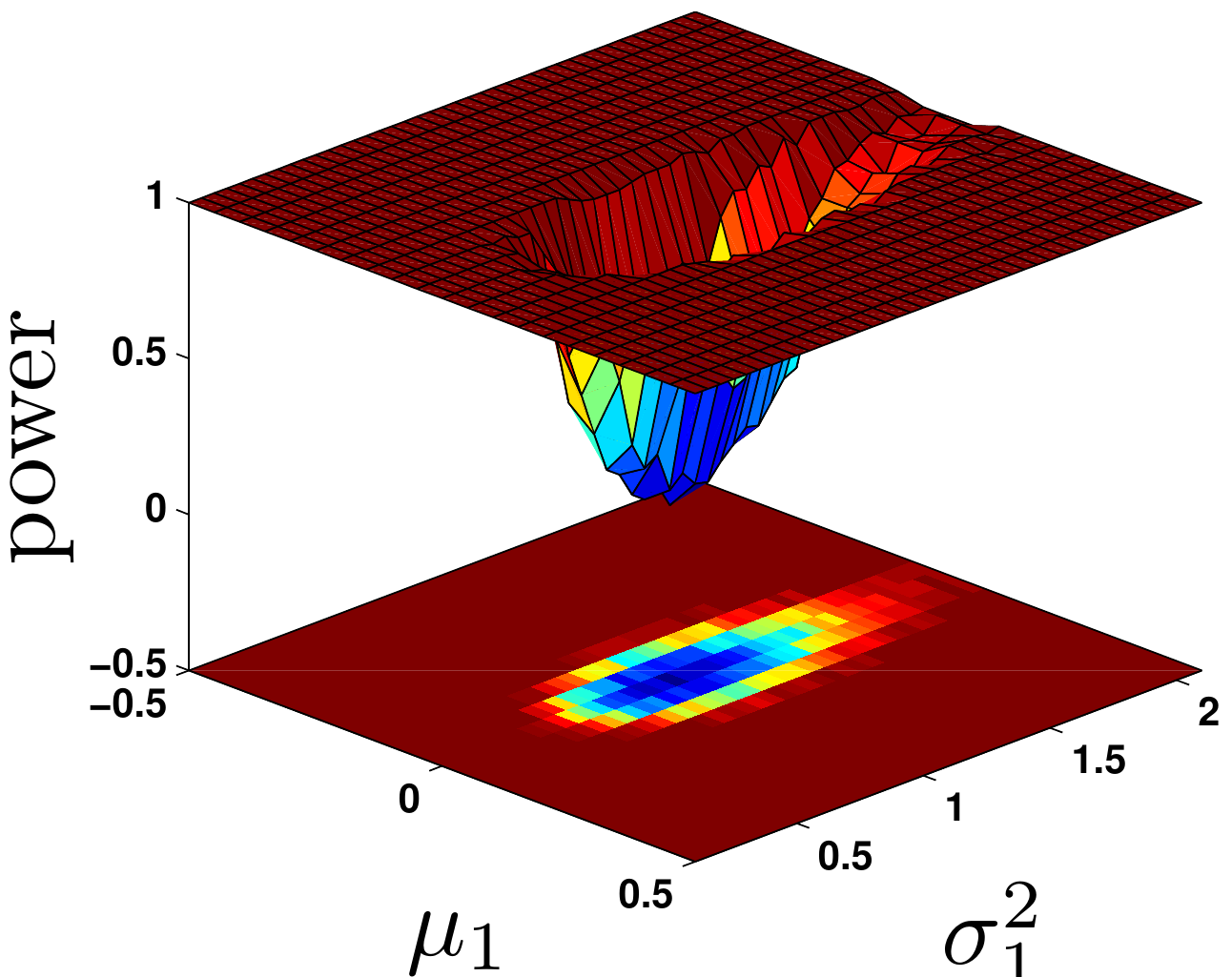}
\label{equalPowerMDt}
}
\subfloat[MD-t: Unbalanced]{
\includegraphics[width=.3\textwidth]{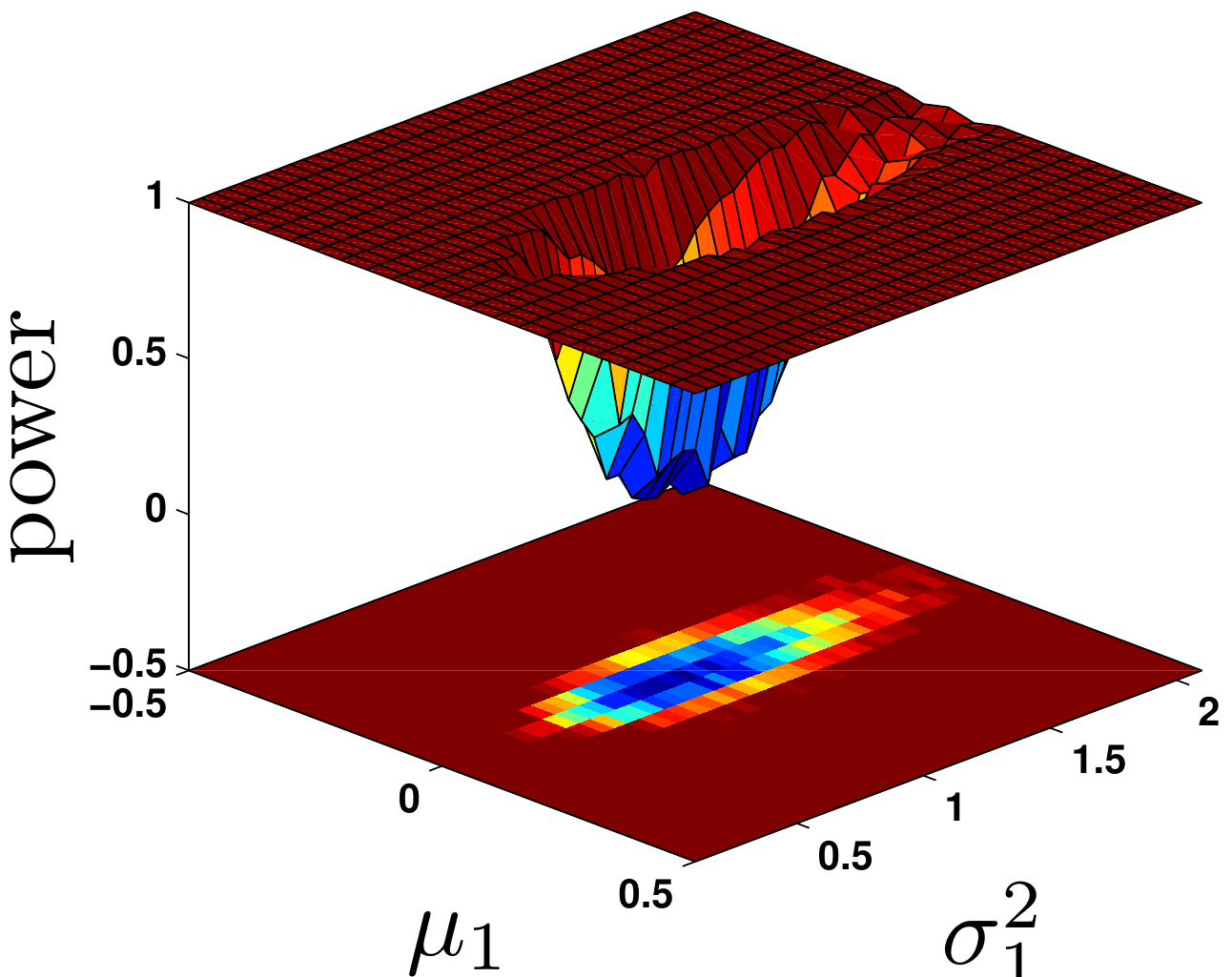}
\label{unequalPowerMDt}
}
\caption{Power surfaces for testing equality of means of the distributions $F_1=N(\mu_1,\sigma_1^2I_d)$ and $F_2=N(0,I_d)$. We see that MD-MD attains the correct level under balanced sample sizes. The MD-t test is not valid for testing equality of means regardless of balanced or unbalanced sample sizes.}
\label{powerMdt}
\end{figure}

Figures \ref{equalPowerMDt} and \ref{unequalPowerMDt} show that under heterogeneous covariances (when $\sigma_1^2 \ne 1$), the MD-t test of equal means does not attain the correct level for either balanced or unbalanced sample sizes.  In the immediate region around $(\mu_1,\sigma_1^2)=(0,1)$, the power of the MD-t test is close to $\alpha$ as expected. However as we move away from $(\mu_1,\sigma_1^2)=(0,1)$, the power quickly increases. Thus if we use the MD-t test for equality of means, we will reject too often. On the other hand this shows that the MD-t test has some power for testing equality of distributions against alternatives where the means are equal but the distributions are not.

\section{Comparison with Other Methods}
\label{power}

In this section we compare DiProPerm to other methods in the simulation contexts described in Table \ref{databank1}. 
First, for testing equality of distributions,  we compare the DiProPerm tests DWD-t and MD-t to the energy test proposed by Szekely and Rizzo \citep{Szekely2004}.  Next, for testing equality of means, we compare the DiProPerm tests DWD-MD and MD-MD to the Random Projection test proposed by Lopes, Jacob and Wainwright \citep{Lopes2011}. Our simulation results show that no test is universally most powerful. As such, our goal is to learn general lessons about the situations under which each method can be expected to do well. 

\begin{table}[h]
\centering
\begin{tabular}{|p{3cm}|p{5cm}|p{5cm}|}
\hline 
Simulation  & Sample 1 & Sample 2 \\ 
\hline 
S1 & $N(0,I_{d})$ & $t(5)^{d}$  \\ 
\hline 
S2 & $N(0,\Sigma_B)$ & $N(\mu,\Sigma_B)$ \\ 
\hline 
S3 & $N( [3,30,0,\ldots,0],I_d)$  $N( [3,-30,0,\ldots,0],I_d)$ & $N( [-3,30,0,\ldots,0],I_d)$  $N( [-3,-30,0,\ldots,0],I_d)$ \\ 
\hline 
\end{tabular} 
\caption{Simulation settings. The notation $N(\mu,\Sigma)$ denotes a multivariate Gaussian distribution with mean $\mu$ and covariance $\Sigma$. In S1, the notation $t(5)^d$ denotes the $d$-variate distribution with iid marginal distribution $t(5)$. In S2, the first 25\% of the coordinates in $\mu$ are zero and the rest are set to $1/\sqrt{n}$. The covariance matrix $\Sigma_B$ has a block structure (described further in the text). In S3, each distribution is an equally weighted Gaussian Mixture of the components listed.}
\label{databank2}
\label{databank1}
\end{table} 

Simulation S1 in Table \ref{databank1} was taken from Szekely and Rizzo. Simulation S2 is a modification of a simulation found in Lopes, Jacob, and Wainwright. Following their simulation setting, we let the covariance matrix $\Sigma_B$ be block-diagonal with identical blocks $B \in R^{5 \times 5}$ along the diagonal. The matrix $B$ has diagonal entries equal to 1 and off-diagonal entries equal to $0.2$. The mean vector is set to the zero vector in sample 1. In the second sample, the mean vector is set to zero in the first 25\% of the coordinates and the rest is set to $1/\sqrt n$. Simulation S3 looks at data arising from equally weighted Gaussian mixtures with the components listed in Table \ref{databank1}. All DiProPerm tests are implemented using $1000$ permutations. Power is estimated through 1000 Monte Carlo simulations at $0.1$ significance level. In Figures \ref{testEqualDist} and \ref{EqualMeans1} we display the power against a range of dimensions.

\subsection{Equality of distributions}
\label{empStudyDist}

The energy statistic is based on the Euclidean distance between pairs of sample elements. The two-sample test statistic is 
\[
\epsilon_{m,n} = \frac{mn}{N} \left( \frac{2}{mn} \sum_{i=1}^m \sum_{j=1}^n ||X_i - Y_j|| - \frac{1}{m^2}  \sum_{i=1}^m \sum_{j=1}^m ||X_i - X_j|| - \frac{1}{n^2}  \sum_{i=1}^n \sum_{j=1}^n ||Y_i - Y_j||  		\right)
\]
The first term measures the average distance between the samples and the last two terms measure the average distance within each sample. The significance of the energy test statistic is assessed using a permutation test. In our implementation of the energy test, we used $1000$ permutations. 

For all simulations in Figure \ref{testEqualDist}, the sample sizes are set to be unbalanced: $m=50, n=150$. Figure \ref{testEqualDist} compares the power of MD-t, DWD-t, and the energy test for testing equality of distributions. The first panel shows the result of simulation S1. The standard Gaussian and $t(5)^d$ both have mean zero but different covariances. Note that the signal in the covariance grows stronger with dimension. In light of this, it is not surprising that the MD-t and DWD-t do not perform as well as the energy test which is more attuned to variance effects. However, as the dimension increases all three tests attain full power. 

The second panel of Figure \ref{testEqualDist} shows the results for simulation S2. All three tests perform well with power increasing to 1 with dimension. Note that the mean effect is along the 45 degree line. The structure of $\Sigma_B$ has the implication that the directions with highest variation are for some constant $c$, $(c,c,c,c,c,0,\ldots,0)$, $(0,0,0,0,0,c,c,c,c,c,0,\ldots,0)$, and etc. Thus the mean effect is further exaggerated by the covariance structure making this a rather unchallenging setting for all three methods.

The result of simulation S3 is shown in the last panel of Figure \ref{testEqualDist}. Here, both the DiProPerm DWD-t and MD-t test are seen to be more powerful than the energy test. This is not surprising since by way of its construction, the energy test can be expected to have difficulty in separating Gaussian mixture data types. The MD-t has good performance but DWD-t has the best power because DWD was developed to handle Gaussian mixture data types.

\begin{figure}[h]
\centering
\subfloat[S1]{
\includegraphics[width=.3\textwidth]{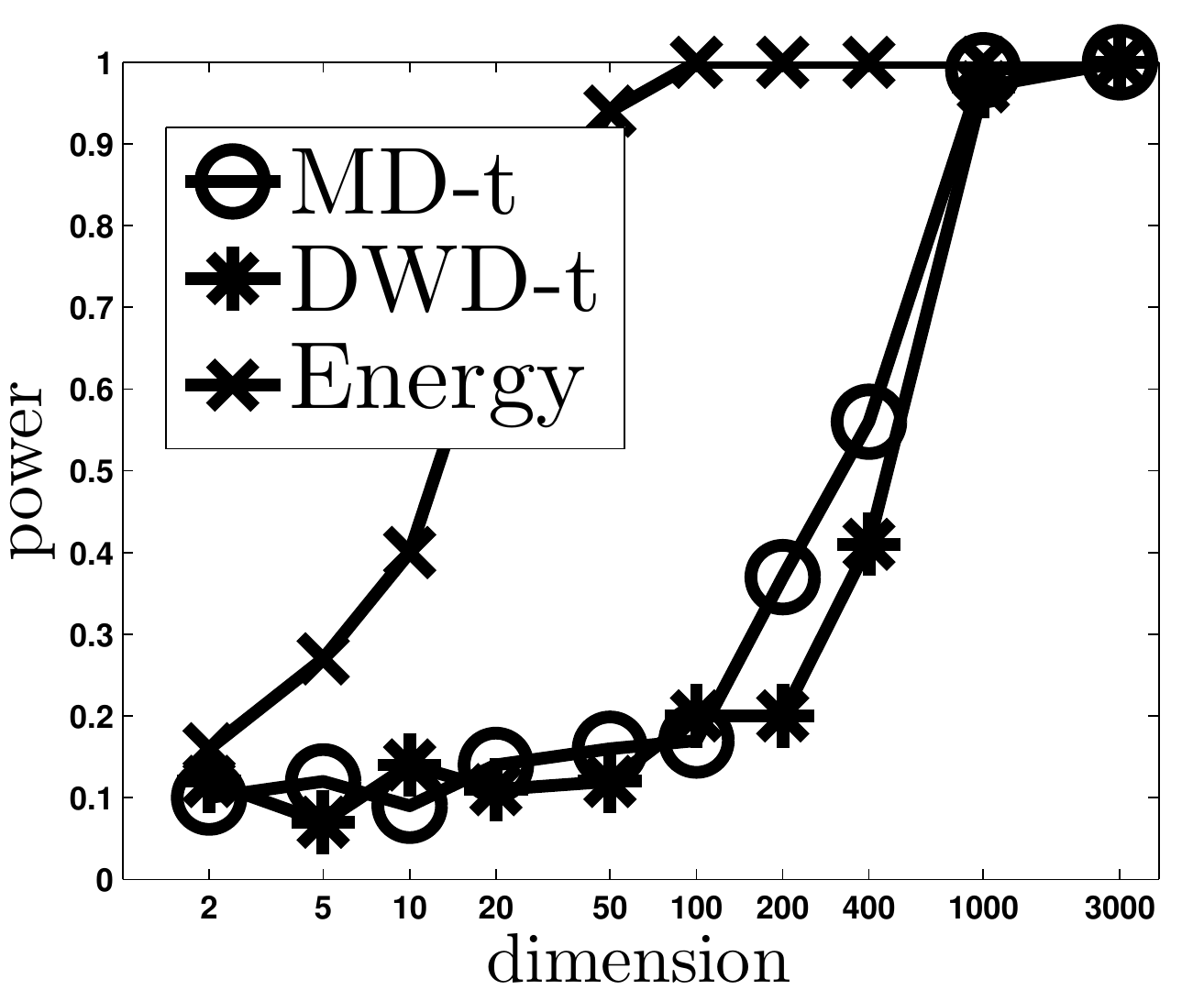}
} 
\subfloat[S2]{
\includegraphics[width=.3\textwidth]{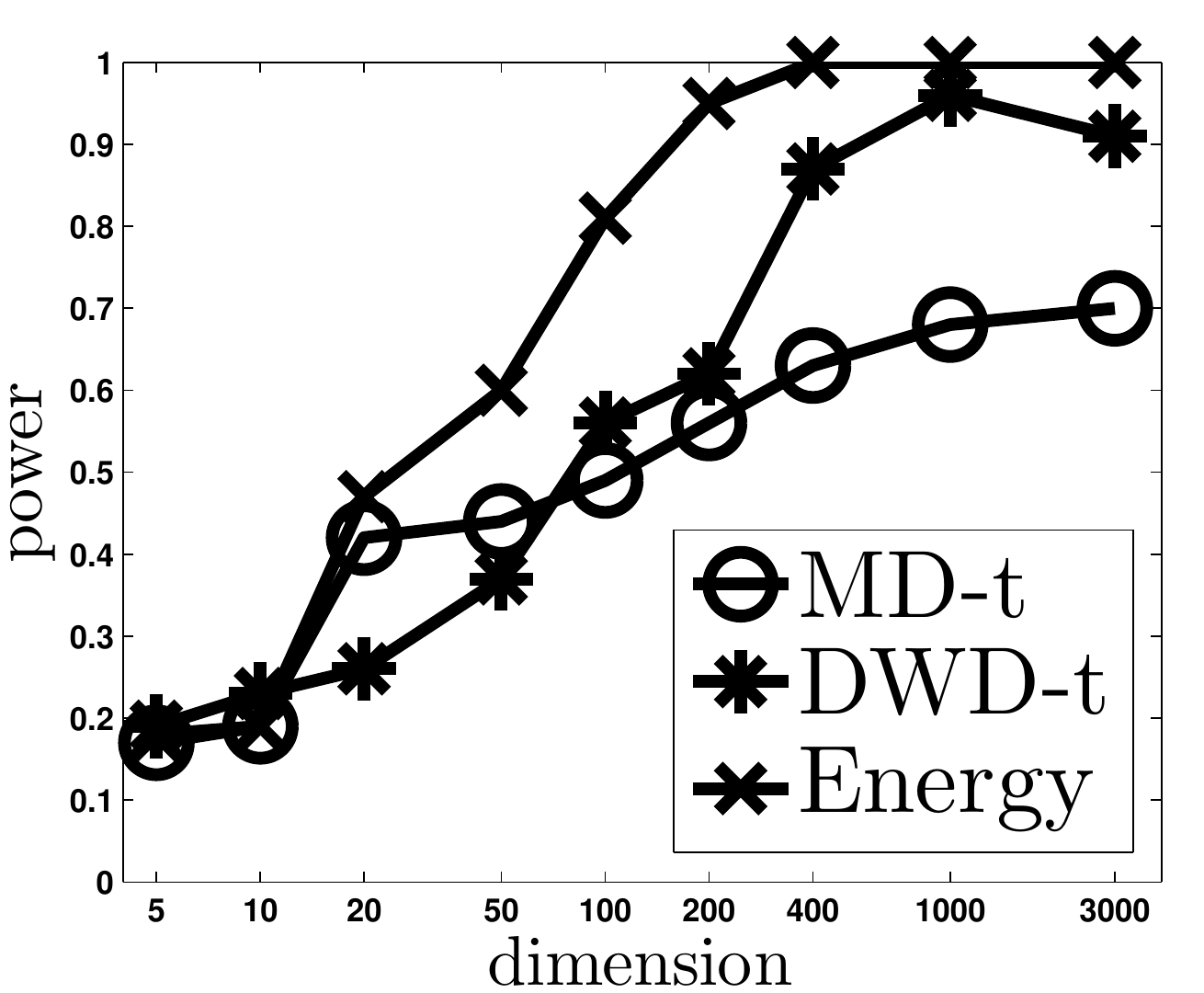}
}
\subfloat[S3]{
\includegraphics[width=.3\textwidth]{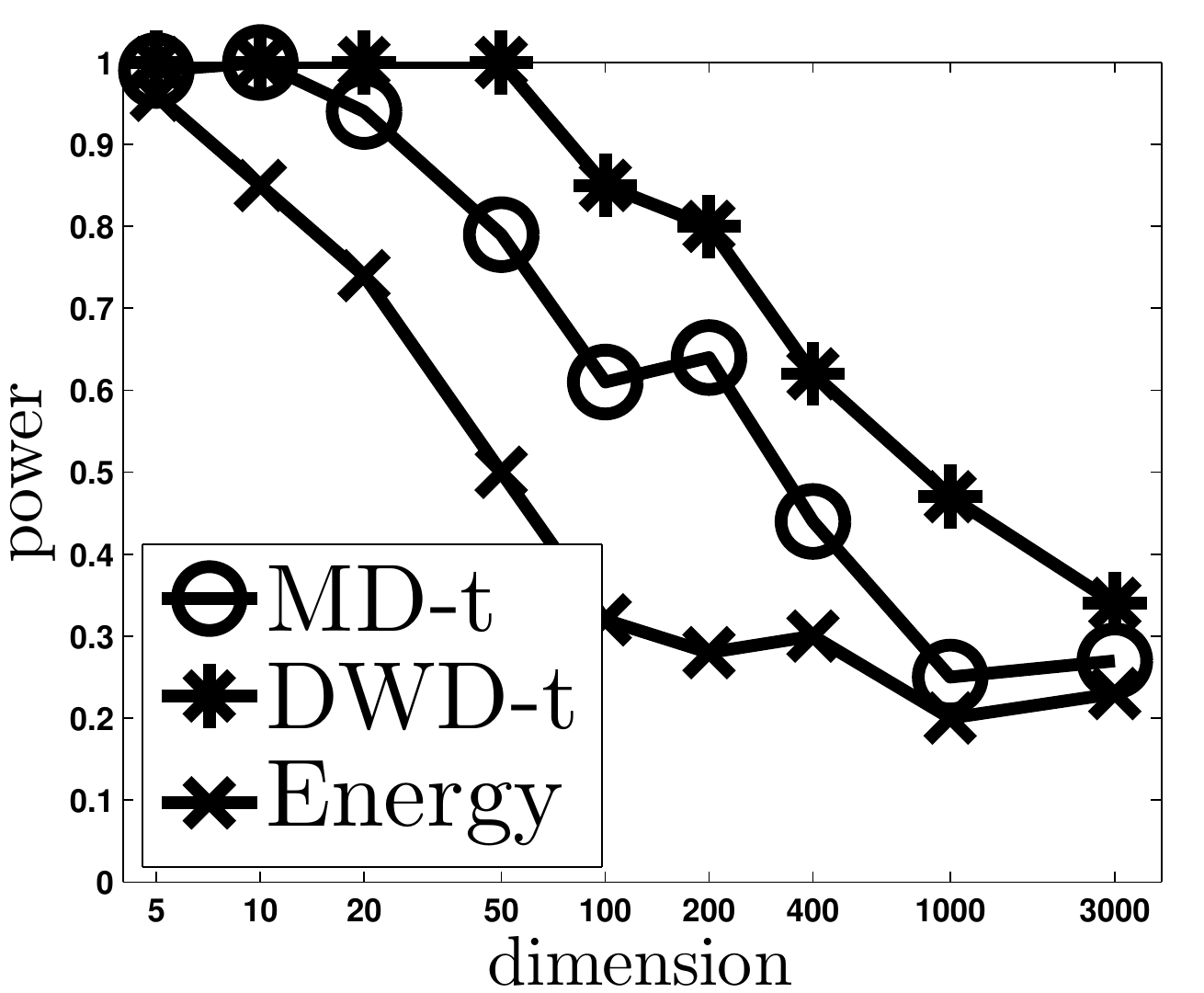}
}
\caption{Power comparison of DWD-t, MD-t and the energy test for testing equality of distributions under the various simulation settings in Table \ref{databank1}.}
\label{testEqualDist}
\end{figure}

\subsection{Equality of means}
\label{empStudyMeans}

In the RP test proposed by Lopes, Jacob and Wainwright, the data is first projected down to a dimension low enough so that the regular Hotelling $T^2$ statistic may be applied \citep{Lopes2011}. The projection matrix is a $k \times d$ matrix with iid $N(0,1)$ entries where $k$ is the dimension of the lower dimensional subspace. In our implementation of the RP method, we follow the authors' recommendation and set the tuning parameter $k = \lfloor{n/2}\rfloor$. The samples are assumed to arise from Gaussian distributions with equal covariances. The resulting statistic then follows an $F$ distribution under the null of equal means.
For all simulations in Figure \ref{EqualMeans1}, the sample sizes are set to be balanced: $m=50, n=50$. The standard multivariate Gaussian and the multivariate $t(5)^d$ both have mean zero, and thus the power of MD-MD and RP should be around $\alpha =0.1$ in simulation S1. The first panel of Figure \ref{EqualMeans1} shows this is indeed the case. Note that if MD-MD or RP were to be used for testing equality of distributions, neither would have power against alternatives such as in S1. 

In simulation S2, the RP method does not perform as well as MD-MD or DWD-MD. This is perhaps due to the DiProPerm tests being able to pick up the mean effect more efficiently than the RP method which tries to sense random directions in very high dimensions. Note that simulation S2 is a setting in which the MD statistic is powerful for either direction DWD or MD as the mean effect is strong. Re-examining Figure \ref{testEqualDist}, we see that the DWD-MD is more powerful than the DWD-t and the MD-MD more powerful than the MD-t for simulation S2. Recall that the covariance structure in S2 amplifies the mean effect. The DiProPerm tests that use the two-sample t-statistic may have lower power than their MD counterpart because the standardization in the t-statistic cancels out some of the effect.

In the Gaussian mixture S3 simulation, the DWD-MD and the RP test are both substantially more powerful than the MD-MD test. In this setting, the direction of discrimination is in the first coordinate direction but the direction of most variation is along the second coordinate. Not surprisingly, MD-MD has trouble in this setting. The RP test, which uses the Mahalanobis distance, is able to correct for this false signal in the second coordinate direction. DWD-MD is seen to perform slightly better than the RP test. Again, DWD is designed to work well in discriminating Gaussian mixture data types and this result matches our expectation. 

\begin{figure}[h]
\centering
\subfloat[S1]{
\includegraphics[width=.3\textwidth]{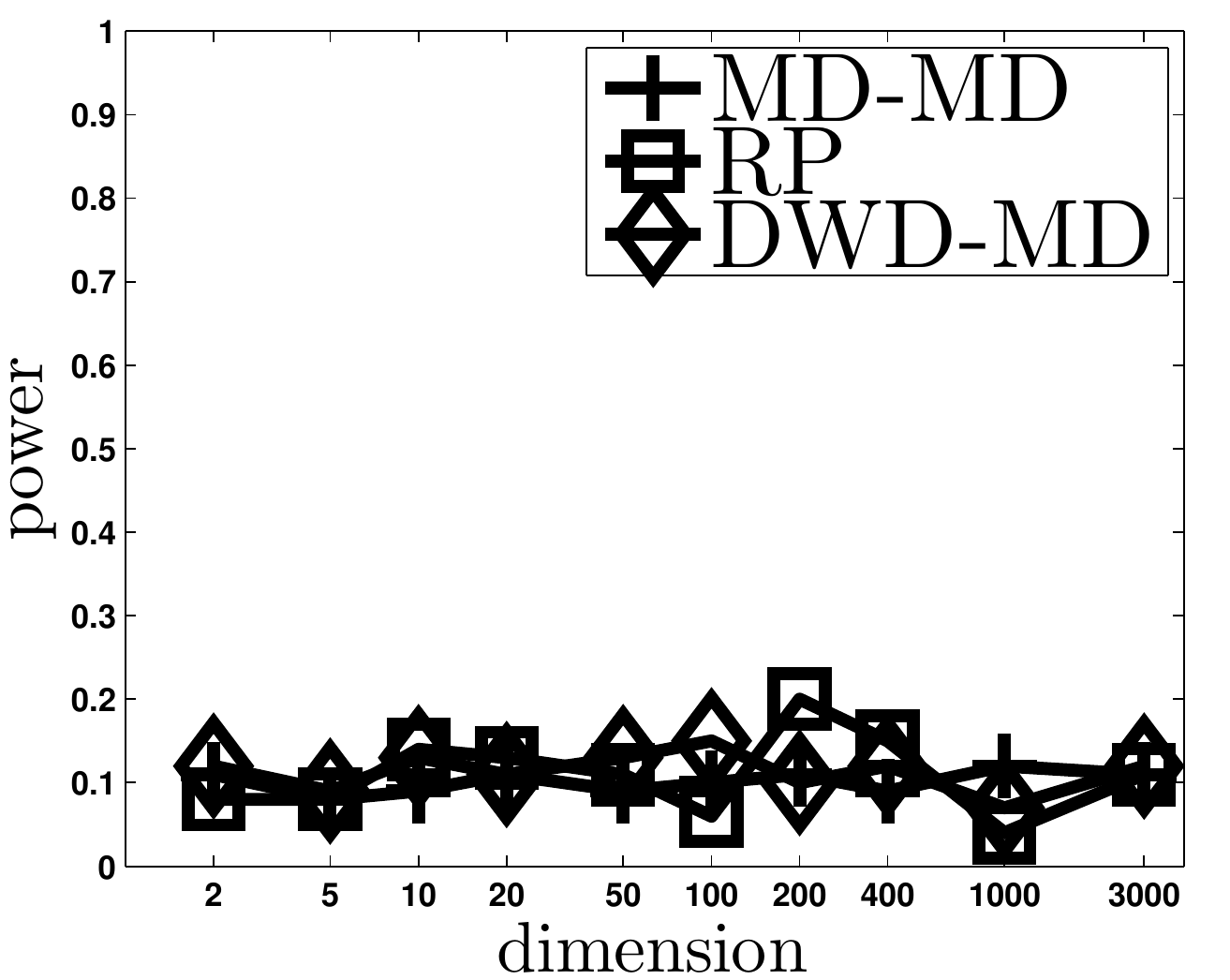}
}
\subfloat[S2]{
\includegraphics[width=.3\textwidth]{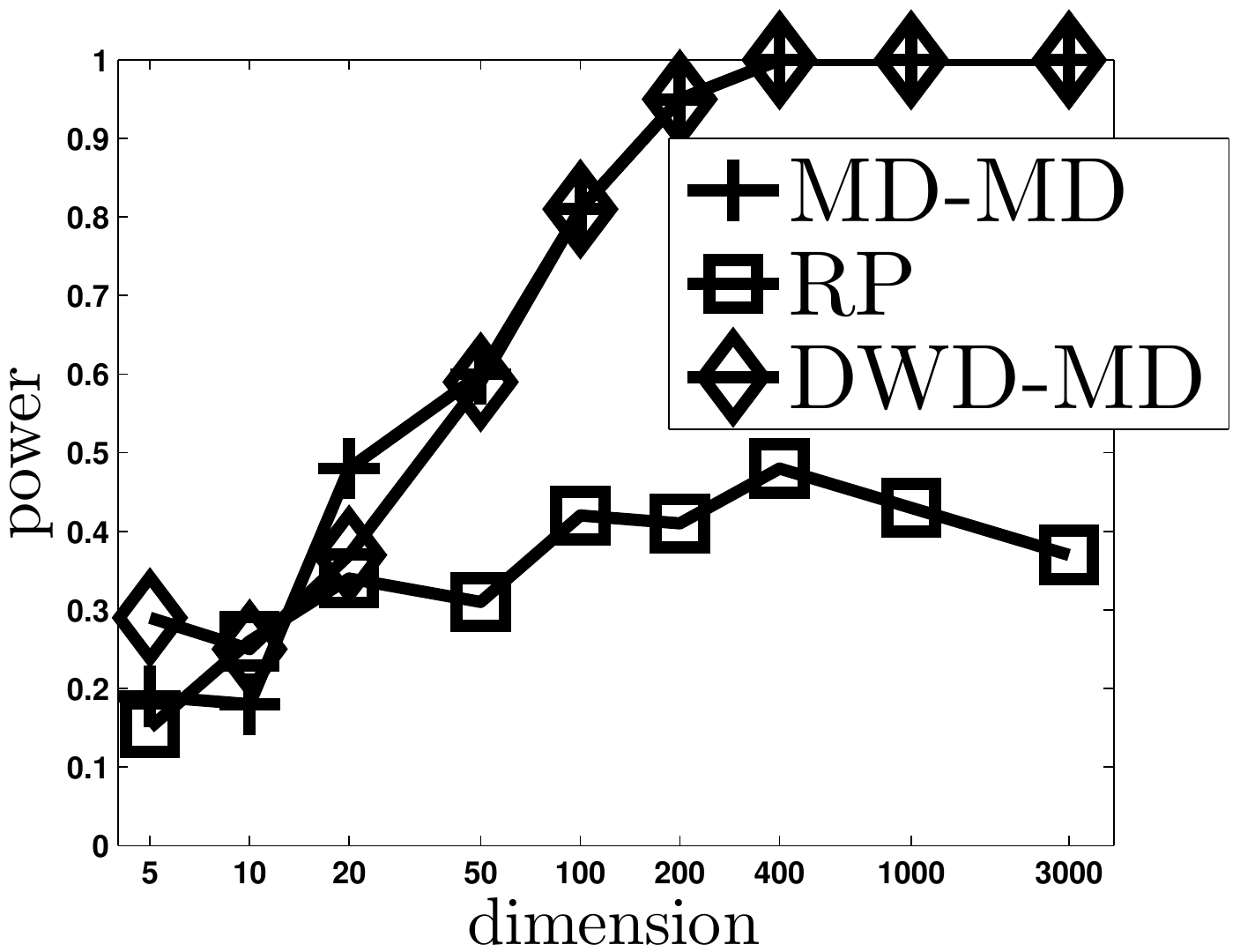}
} 
\subfloat[S3]{
\includegraphics[width=.3\textwidth]{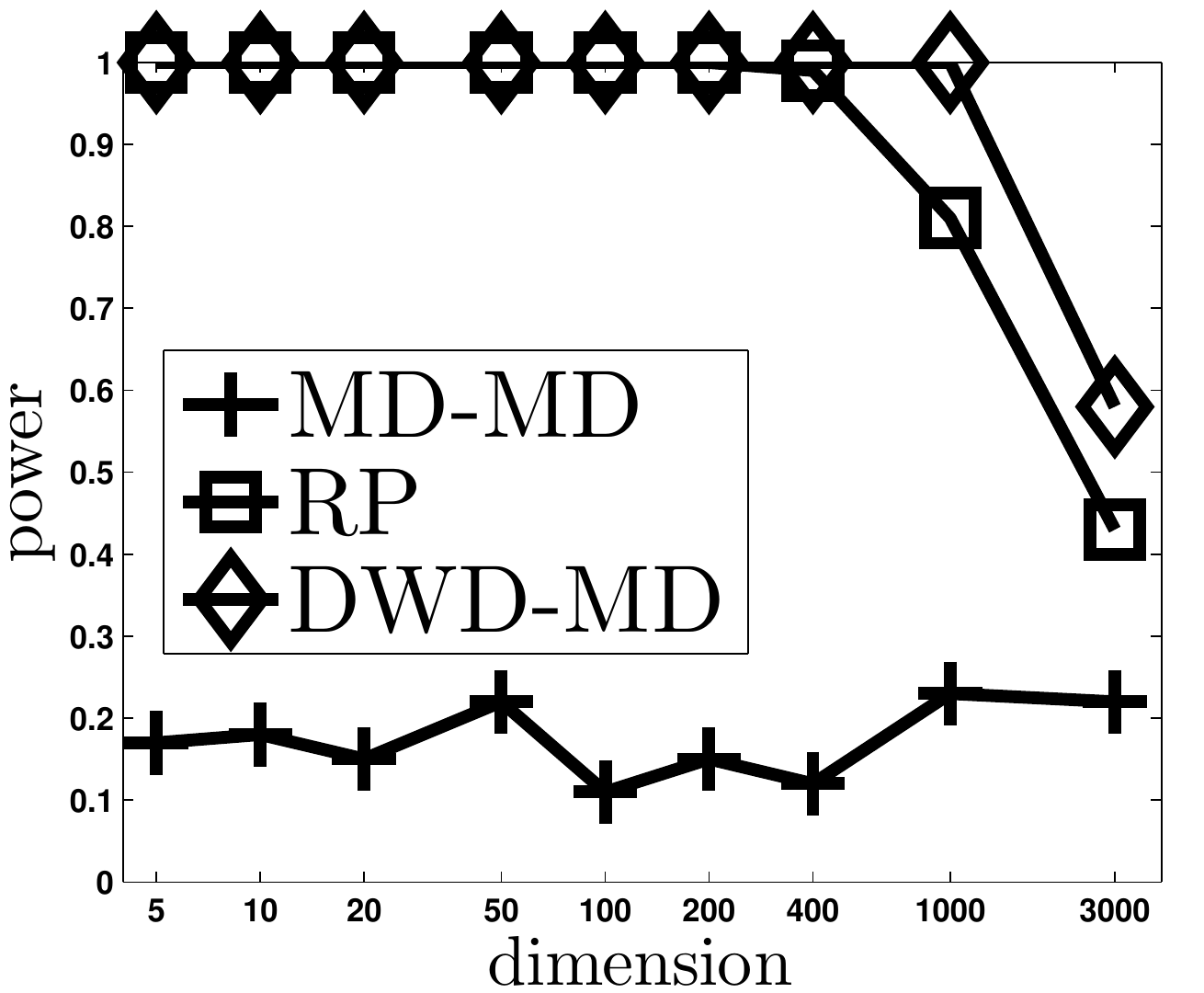}
}
\caption{Power comparison of MD-MD versus the RP method for testing equality of means under the various simulation settings in Table \ref{databank1}.}
\label{EqualMeans1}
\label{EqualMeans2}
\end{figure}

\section{Application: Microarray data analysis }
\label{App}

The first application of DiProPerm to a real dataset can be found in \cite{Wichers2007}. DiProPerm was applied to an HDLSS dataset and used to find a statistically significant difference between heart rates of rats among different treatment groups. In this section we will apply DiProPerm to a different kind of HDLSS data --- gene expression microarray data.

Two HDLSS datasets are examined. The first dataset is denoted UNCGEO and the second UNCUP, following the naming convention of their source which can be found at \verb|http://peroulab.med.unc.edu/|.
 The UNCGEO datasets consists of gene expression data of 9674 genes measured on 50 breast cancer patients at UNC. The UNCUP dataset looks at the same set of genes measured on 80 breast cancer patients in another study at UNC.
We performed many different hypotheses of interest within each dataset. We highlight two particular comparisons here which highlight the main point that formal hypothesis testing is an important component of visualization in high dimensions. 

The UNCGEO patients are divided into standard breast cancer subtypes: 1) Luminal A versus 2) Luminal B and the UNCUP data into the groups: 1) Luminals (Luminal A and Luminal B) versus 2) HER and Basal. Luminals have a very different gene expression signature from HER and Basal. On the other hand, the difference between Luminal A and Luminal B is less clear cut. For each dataset, we use DWD-t to test equality of distributions between the gene expression in group 1 and group 2. Note that we have a HDLSS setting here since the number of genes well exceeds the sample sizes in each subgroup.

Figure \ref{projMicroarray} shows the data projected onto DWD directions. The projections in the left panel do not overlap at all whereas the projections in the right panel have a small amount of overlap.  These projection plots suggest that the separation is better for Luminal A vs. Luminal B in the UNCGEO dataset than for Luminals vs. HER \& Basal in the UNCUP dataset. However as previously seen in the toy example in Section \ref{ACE}, great care is needed before drawing conclusions of this type.

\begin{figure}[h!]
\centering
\includegraphics[width=.75\textwidth]{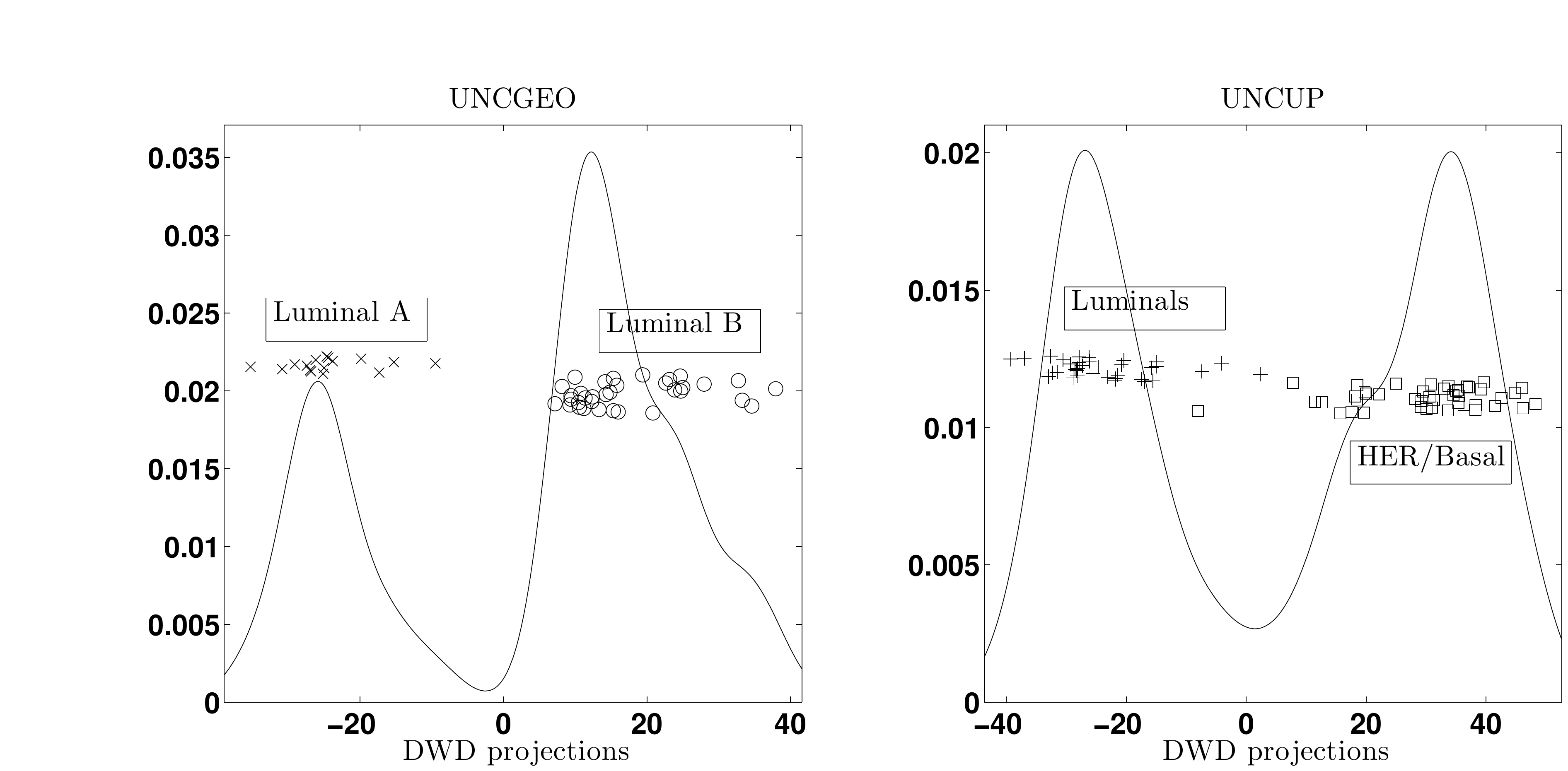}
\caption{ One dimensional projection plots onto DWD directions for the UNCGEO dataset and the UNCUP dataset. The separation in the projection plot for the UNCGEO dataset is more visually pronounced than in the UNCUP dataset. We will rigorously assess this visual result using DiProPerm.}
\label{projMicroarray}
\end{figure}

Figure \ref{permMicroarray} displays the DiProPerm test results.  Each dot represents the test statistic resulting from a single permutation in the permutation test. We mark the position of the original univariate t-statistic with a vertical dashed line. The empirical p-values show the difference in the UNCGEO dataset is not significant while the difference in the UNCUP dataset is very significant. (We also display the Guasisan fit p-value and Gaussian fit z-score, two other types of ``p-values" described in Section \ref{Perm} of the Supplement). This result on a real world dataset parallels what we saw on the simulated toy dataset in Section \ref{ACE} --- what may seem to be a visually striking separation in lower dimensional visualizations could well be an artifact of over-fitting or sampling variation.

\begin{figure}[h]
\centering
\includegraphics[width=.75\textwidth]{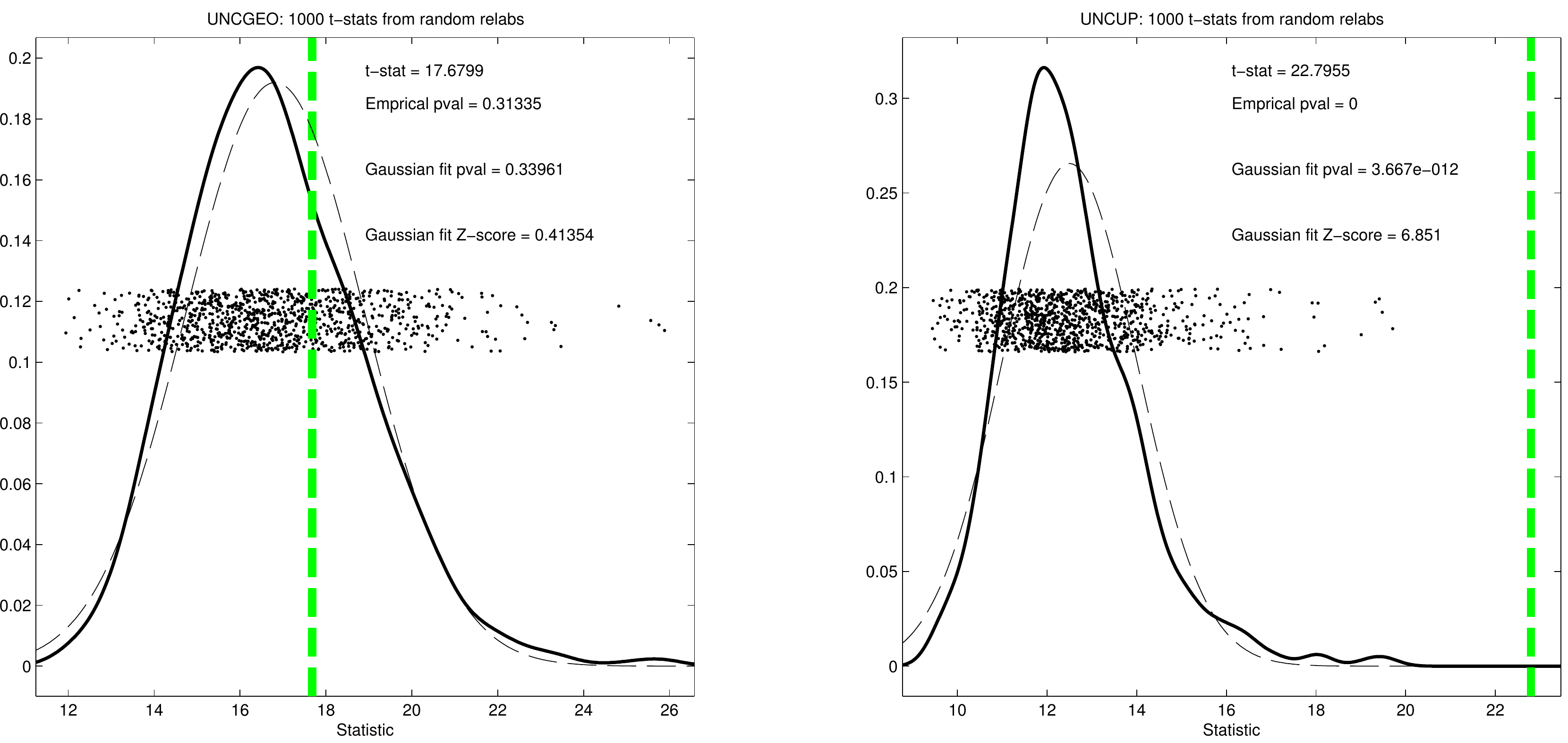}
\caption{DWD-t test result for the UNCGEO (left) and UNCUP (right) datasets. In the UNCGEO study (left), the difference between the Luminal A and Luminal B subgroups is not significant. In the UNCUP study (right), the difference between the Luminals and HER \& Basal subgroups is very significant. This is surprising because the projection plots in Figure \ref{projMicroarray} suggest the contrary.}
\label{permMicroarray}
\end{figure}

\section{Matlab Software}
Matlab software for DiProPerm is available at
\verb|http://www.unc.edu/~marron/marron_software.html|.

\noindent
\section*{Acknowledgements} 
The work presented in this paper was supported in part by the NSF Graduate Fellowship, and NIH grant T32 GM067553-05S1.

\bibliography{master_bibtex}

\newpage
\appendix

\section{Proofs}

\begin{lem}
Let $X_1, \ldots, X_m$ be a sample from the d-variate Gaussian distribution $N(\mu_x, \sigma_x^2 I_d)$ and $Y_1,\ldots,Y_n$ be an independent sample from the d-variate Gaussian distribution $N(\mu_y,\sigma_y^2 I_d)$ where $\sigma_x^2 \ne \sigma_y^2$. Let  $\tilde X_k = X_k' (\bar X - \bar Y)$. Let $\overline{\tilde X_{1:k-1}}$ be the sample mean of $\tilde X_1,\ldots \tilde X_{k-1}$. Under $\mu_x = \mu_y$, we have, for $k=2,\ldots, m$ 
\[
\frac{d^{-1/2} (  (\tilde{X}_{k} - \overline{\tilde X_{1:k-1}}) )}{\{\frac{k}{k-1} \sigma_x^2 (\sigma_x^2/m + \sigma_y^2/n)\}^{1/2}} \convd N(0,1) \text{ as } d \to \infty.
\]
Similarly we have 
\[
\frac{d^{-1/2} ( (\tilde{Y}_{k} - \overline{\tilde Y_{1:k-1}}) )}{\{\frac{k}{k-1} \sigma_y^2 (\sigma_x^2/m + \sigma_y^2/n)\}^{1/2}} \convd N(0,1) \text{ as } d \to \infty
\]
$k=2,\ldots, n$.
\label{projDiff}
\end{lem}

\begin{proof}
We can write $\tilde{X}_{k} - \overline{\tilde X_{1:k-1}} $ as a sum of products
\begin{align}
\tilde{X}_{k} - \overline{\tilde X_{1:k-1}}  = \sum_{p=1}^d (X_k - \bar X_{1:k-1})^{(p)} (\bar X - \bar Y)^{(p)}
\label{componentDef}
\end{align}
where  $X^{(p)}$ simply refers to the $p$-th component in the $d$-dimensional vector $X$. The expectation of the summands in \eqref{componentDef} is zero:
\begin{align*}
E  (X_k - \bar X_{1:k-1})^{(p)} (\bar X - \bar Y)^{(p)} &= E (X_k^{(p)} \bar X^{(p)}) - E(\bar X_{1:k-1}^{(p)} \bar X^{(p)}) \\
&- E(X_k^{(p)} \bar Y^{(p)}) + E(\bar X_{1:k-1}^{(p)} \bar Y^{(p)}) \\
&=0
\end{align*}

Next we look at the variance of the summands. Recall for Gaussian data, zero covariance is equivalent to independence. We know the covariance between $(X_k - \bar X_{1:k-1})^{(p)}$ and $(\bar X - \bar Y)^{(p)}$ is zero since the expectation of the latter is zero and the expectation of the product was shown above to be zero as well. Thus each summand in \eqref{componentDef} is the product of two independent variables. The variance of a product of independent variables (see \ref{varProd} for a derivation), $U$ and $V$, is 
\begin{equation}
(EU)^2 Var(V) + (EV)^2 Var(U) + Var(U) Var(V).
\label{prodVar}
\end{equation}

Thus we have
\begin{align*}
Var  (X_k - \bar X_{1:k-1})^{(p)} (\bar X - \bar Y)^{(p)} &= Var  (X_k - \bar X_{1:k-1})^{(p)}Var (\bar X - \bar Y)^{(p)} \\
&= \frac{k}{k-1} \sigma_x^2 (\sigma_x^2/m + \sigma_y^2/n)
\end{align*}
By the Central Limit Theorem, we have
\[
\frac{d^{1/2} ( \frac{1}{d} (\tilde{X}_{k} - \overline{\tilde X_{1:k-1}}) )}{\{\frac{k}{k-1} \sigma_x^2 (\sigma_x^2/m + \sigma_y^2/n)\}^{1/2}} \convd N(0,1) \text{ as } d \to \infty
\]
\end{proof}

\begin{lem}
Let $X_1, \ldots, X_m$ be a sample from the d-variate Gaussian distribution $N(\mu_x, \sigma_x^2 I_d)$ and $Y_1,\ldots,Y_n$ be an independent sample from the d-variate Gaussian distribution $N(\mu_y,\sigma_y^2 I_d)$ where $\sigma_x^2 \ne \sigma_y^2$. Let $\pi$ be a permutation of $\{1,\ldots,N=m+n\}$. Let $\bar Z_\pi = ( \bar Z_{\pi(1:m)} - \bar Z_{\pi(m+1:N)}  )$ be the MD direction trained on the permuted labels determined by $\pi$. We have for $i=1,\ldots,m$, 
\[
E(({Z}_{\pi(i)} - \overline{Z_{\pi(1:m)}})^{(k)} \bar Z_\pi^{(k)}) 
\]
is non-zero.
Similarly, for $i=m+1,\ldots,N$, we have 
\[
E( ({Z}_{\pi(i)} - \overline{Z_{\pi(m+1:N)}})^{(k)} \bar Z_\pi^{(k)})
\]
is non-zero.

\label{nonzero}
\end{lem}

\begin{proof}
We prove the first statement. The second can be shown in a similar fashion. Let $P(n,k)$ denote the number of $k$ permutations of $n$, i.e.
\[
P(n,k) = n \cdot (n-1) \cdot (n-2) \cdots (n-k+1)
\]
We have for $i=1,\ldots,m$ and $k=1,\ldots,d$,
\begin{align*}
E(({Z}_{\pi(i)} - \overline{Z_{\pi(1:m)}})^{(k)} \bar Z_\pi^{(k)})  &= E ((Z_{\pi(i)} - \bar Z_{\pi(1:m)})^{(k)} ( \bar Z_{\pi(1:m)} -  \bar Z_{\pi(m+1:N)} )^{(k)}) \\
&= EZ_{\pi(i)}^{(k)} \bar Z_{\pi(1:m)}^{(k)} - EZ_{\pi(i)}^{(k)} \bar Z_{\pi(m+1:N)}^{(k)} - E( \bar Z_{\pi(1:m)}^{(k)} )^2 +  E \bar Z_{\pi(1:m)}^{(k)} \bar Z_{\pi(m+1:N)}^{(k)} \\
&= \frac{(EZ_{\pi(i)}^{(k)})^2}{m} +\frac{m}{m-1} \mu^2 - \mu^2 - E( \bar Z_{\pi(1:m)}^{(k)} )^2 + \mu^2\\
&= \frac{ var(Z_{\pi(i)}^{(k)}) + \mu^2}{m} + \frac{m}{m-1} \mu^2 - ( var (\bar Z_{\pi(1:m)}^{(k)} ) + \mu^2) \\
&= \frac{ var(Z_{\pi(i)}^{(k)}) }{m} - var (\bar Z_{\pi(1:m)}^{(k)} ) \\
&=\frac{m}{N} \{ \frac{\sigma_x^2}{m} - \frac{1}{m^2} \frac{1}{w_1}   \sum_{r=0}^{m-1} P(m-1,r) P(n,m-r) [r \sigma_x^2 + (m-r)\sigma_y^2] \} \\ 
&+\frac{n}{N} \{ \frac{\sigma_y^2}{m} - \frac{1}{m^2}  \frac{1}{w_2}   \sum_{r=0}^{n-1} P(n-1,r) P(m,m-r)  [r \sigma_y^2 + (m-r)\sigma_x^2] \} 
\end{align*}
where $w_1$ and $w_2$ are the weights
\[
w_1 :=   \sum_{r=0}^{m-1} P(m-1,r) P(n,m-r)  \quad \text{ and } \quad w_2:= \sum_{r=0}^{n-1} P(n-1,r) P(m,m-r)
\]
Thus if $\sigma_x^2 \ne \sigma_y^2$, we have $E(({Z}_{\pi(i)} - \overline{Z_{\pi(1:m)}})^{(k)} \bar Z_\pi^{(k)})$ is nonzero. 
\end{proof}

\begin{lem}
Let $Z_1,Z_2$ be two random variables in $\mathbb R^d$ such that $Z_1^{(k)}Z_2^{(k)}$ are i.i.d. for $k=1,\ldots,d$ and $E(Z_1^{(k)}Z_2^{(k)})$ exists and is finite. Then 
\[
\frac{1}{d^2} ( Z_1 \cdot Z_2 )^2  \to [E(Z_1^{(k)}Z_2^{(k)})]^2 \text{ in probability}
\]
\label{dot}
\end{lem}
\begin{proof} 
By the Law of Large Numbers, we have
\[
\frac{ 1}{d} (Z_1 \cdot Z_2) \to E(Z_1^{(k)}Z_2^{(k)}) \text{ in probability.}
\]
By Continuous Mapping Theorem, we have
\[
\frac{ 1}{d^2} (Z_1 \cdot Z_2)^2 \to [E(Z_1^{(k)}Z_2^{(k)})]^2 \text{ in probability.}
\]
\end{proof}

Now we have all the necessary ingredients to prove Theorem \ref{MDtTheorem} in Section \ref{valMDt}. 

\begin{proof}[Proof of Theorem \ref{MDtTheorem}]

To prove the first part of Theorem \ref{MDtTheorem}, we decompose $s_{\tilde X}^2$ and $s_{\tilde Y}^2$ into a sum of independent variables. Let $\overline{\tilde{X}}_{k-1}$ be the sample mean of the first $k-1$ projections $\tilde X_1,\ldots \tilde X_{k-1}$. We will write $s_{\tilde X}^2$ in a recursive fashion. Define $s_1^2 = 0$. We will use the following recursive formula to define $s_k^2$ for $k = 2,\ldots,m$
\begin{equation}
(k-1) s_k^2 = (k-2) s_{k-1}^2 + \frac{k-1}{k} ( \tilde{X}_k - \overline{\tilde{X}}_{k-1})^2
\label{recur}
\end{equation}
Since $s_{k-1}^2$ is independent of $( \tilde{X}_k - \overline{\tilde{X}}_{k-1})^2$, this recursive viewpoint allows us to decompose $s_{\tilde X}^2 = s_m^2$ into a sum of independent terms. Using the result in Lemma \ref{projDiff} and the second-order Delta method, we have
\begin{equation}
\frac{ \frac{1}{d} (\tilde{X}_{k} - \overline{\tilde X_{1:k-1}})^2 }{\frac{k}{k-1} \sigma_x^2 (\sigma_x^2/m + \sigma_y^2/n)}  \convd \chi^2(1) \text{ as } d \to \infty
\label{inc}
\end{equation}
Inputting expression \eqref{inc} into the recursion defined in \eqref{recur} and exploiting the independence of the individual terms in $s_{\tilde{X}}^2$, we get
\begin{equation*}
\frac{1}{d} s_{\tilde{X}}^2  \convd \frac{1}{m-1} \sigma_x^2 (\frac{\sigma_x^2}{m} + \frac{\sigma_y^2}{n}) \chi^2(m-1)  \text{ as } d \to \infty
\end{equation*}
Similarly, we can show for the sample of projections $\tilde Y_1,\ldots, \tilde Y_n$,
\begin{equation*}
\frac{1}{d} s_{\tilde{Y}}^2  \convd \frac{1}{n-1} \sigma_y^2 (\frac{\sigma_x^2}{m} + \frac{\sigma_y^2}{n}) \chi^2(n-1)  \text{ as } d \to \infty
\end{equation*}
Thus we have
\begin{align*}
\frac{1}{d} S_{m,n}(Z) &= \frac{1}{d} \left( \frac{s_{\tilde{X}}^2 }{m} + \frac{s_{\tilde{Y}}^2 }{n} \right) \\
& \convd \frac{1}{m-1} \frac{\sigma_x^2}{m} (\frac{\sigma_x^2}{m} + \frac{\sigma_y^2}{n}) \chi^2(m-1) +  \frac{1}{n-1} \frac{\sigma_y^2}{n} (\frac{\sigma_x^2}{m} + \frac{\sigma_y^2}{n}) \chi^2(n-1)
\end{align*}

For the second part in Theorem \ref{MDtTheorem}, we expand the sample variance of the projected values in the permuted group as follows:
\begin{align*}
s_{{\tilde Z_{\pi(1:m)}}}^2 &= \frac{1}{m-1} \sum_{i=1}^m (\tilde{Z}_{\pi(i)} - \overline{\tilde Z_{\pi(1:m)}})^2 \\
&= \frac{1}{m-1} \sum_{i=1}^m ( ({Z}_{\pi(i)} - \overline{Z_{\pi(1:m)}}) \cdot ( \bar Z_{\pi(1:m)} - \bar Z_{\pi(m+1:N)}  )  )^2 \\
\end{align*}
Lemma \ref{nonzero} shows $E (Z_{\pi(i)} - \bar Z_{\pi(1:m)})^{(k)} ( \bar Z_{\pi(1:m)} -  \bar Z_{\pi(m+1:N)} )^{(k)}$ is nonzero. Now apply Lemma \ref{dot} with $Z_1 = ({Z}_{\pi(i)} - \overline{Z_{\pi(1:m)}}) $ and $Z_2 = ( \bar Z_{\pi(1:m)} - \bar Z_{\pi(m+1:N)}  )$ to see that $\frac{1}{d^2} s_{{\tilde Z_{\pi(1:m)}}}^2 $ converges in probability to a nonzero constant. A similar argument can be applied to $s_{{\tilde Z_{\pi(m+1:N)}}}^2$. Combining these results, it immediately follows that $\frac{1}{d^2} S_{m,n}(Z_\pi)$ converges in probability to a nonzero constant.

\end{proof}

\section{MD-scaled MD}
\label{MDsMD}
We established in Section \ref{valMDMD} that under certain conditions, the MD-MD test is valid when sample sizes are balanced. Under these same conditions, MD-MD is no longer a valid test however when sample sizes are unbalanced. Here we propose a modification of MD-MD, called MD-scaled MD, that is asymptotically valid, as $m,n \to \infty$ for fixed $d$, for equality of means when covariances are unequal and sample sizes are unbalanced. 

We have chosen the classical asymptotic regime here to take advantage of the following results. Janssen proved the permutation test for equality of means based on the studentized statistic, 
\begin{equation}
{ m^{1/2} (\bar X - \bar Y)}/{\{ s_x^2 + \frac{m}{n} s_y^2\}^{1/2} }
\label{MDstud}
\end{equation} 
where $s_x^2$ and $s_y^2$ are the standard unbiased estimators of $\sigma_x^2$ and $\sigma_y^2$, is asymptotically valid as $m,n \to \infty$ for the univariate case \citep{Janssen1997}. Janssen's result easily extends to the multivariate case if we assume a spherical covariance structure. Let $X_1, \ldots, X_m$ be a sample from a d-variate distribution with mean and covariance $(\mu_X, \sigma_x^2 I_d)$ and $Y_1,\ldots,Y_n$ be an independent sample with mean and covariance $(\mu_Y,\sigma_y^2 I_d)$. We propose the MD-scaled MD DiProPerm test whereby the MD direction is used in Step 1 of DiProPerm and a scaled MD statistic as in Equation \eqref{MDstud} is used in Step 2. The MD-scaled MD statistic is
\begin{equation}
{T_{m,n}(Z)}/{\{ {s_x^2}/{m} + {s_y^2}/{n} \}^{1/2}}
\label{MDsMDeq}
\end{equation}
where $T_{m,n}(Z)$ is as in Section \ref{valMDMD}. The asymptotic validity of the MD-scaled MD statistic (as $m,n \to \infty$) follows immediately from Janssen's result. Note that normality is not an assumption here. 

\begin{figure}[h]
\centering
\subfloat[MD-MD: Unbalanced]{
\includegraphics[width=.3\textwidth]{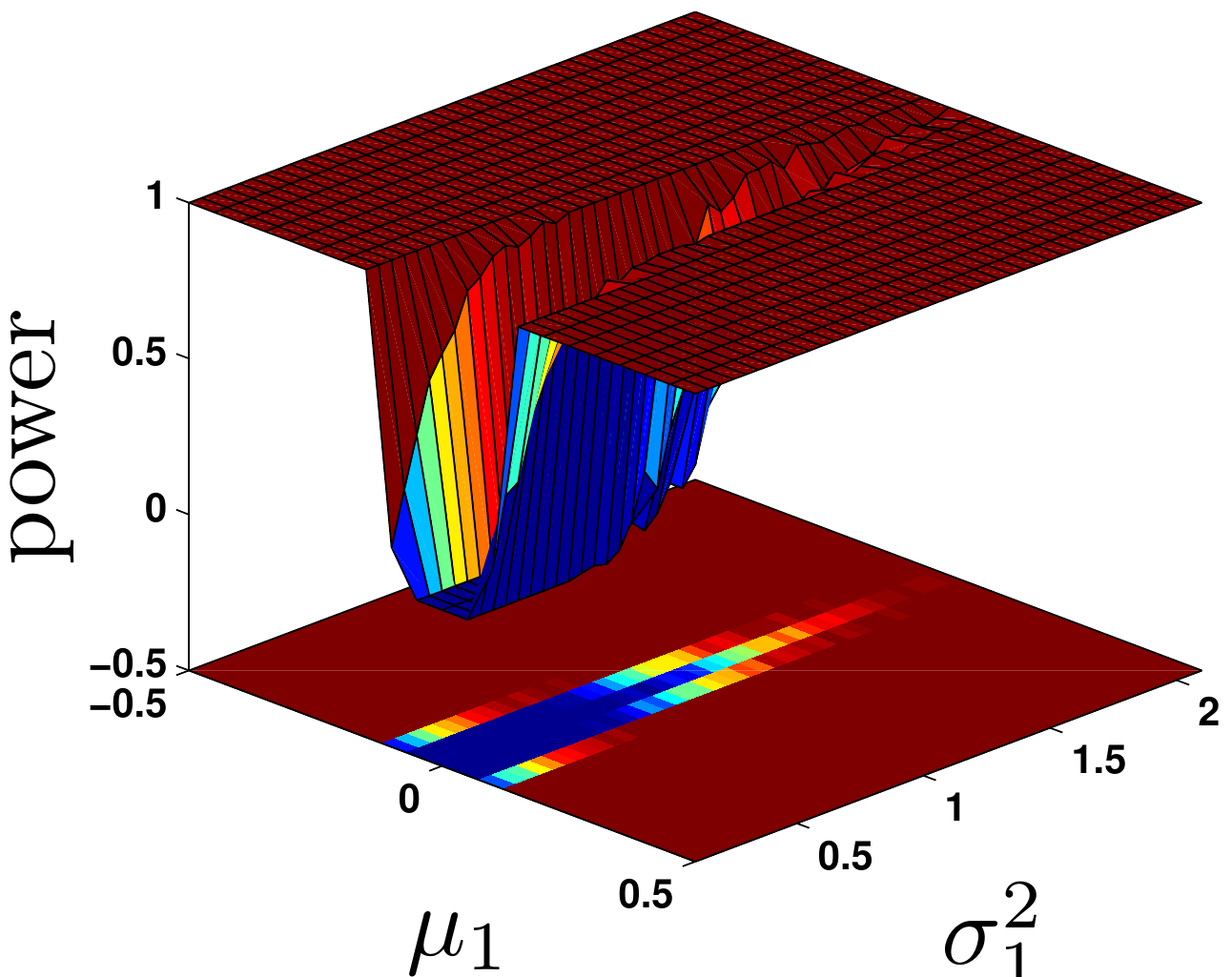}
\label{unequalPowerMDMD}
}
\subfloat[MD-scaled MD: Unbalanced]{
\includegraphics[width=.3\textwidth]{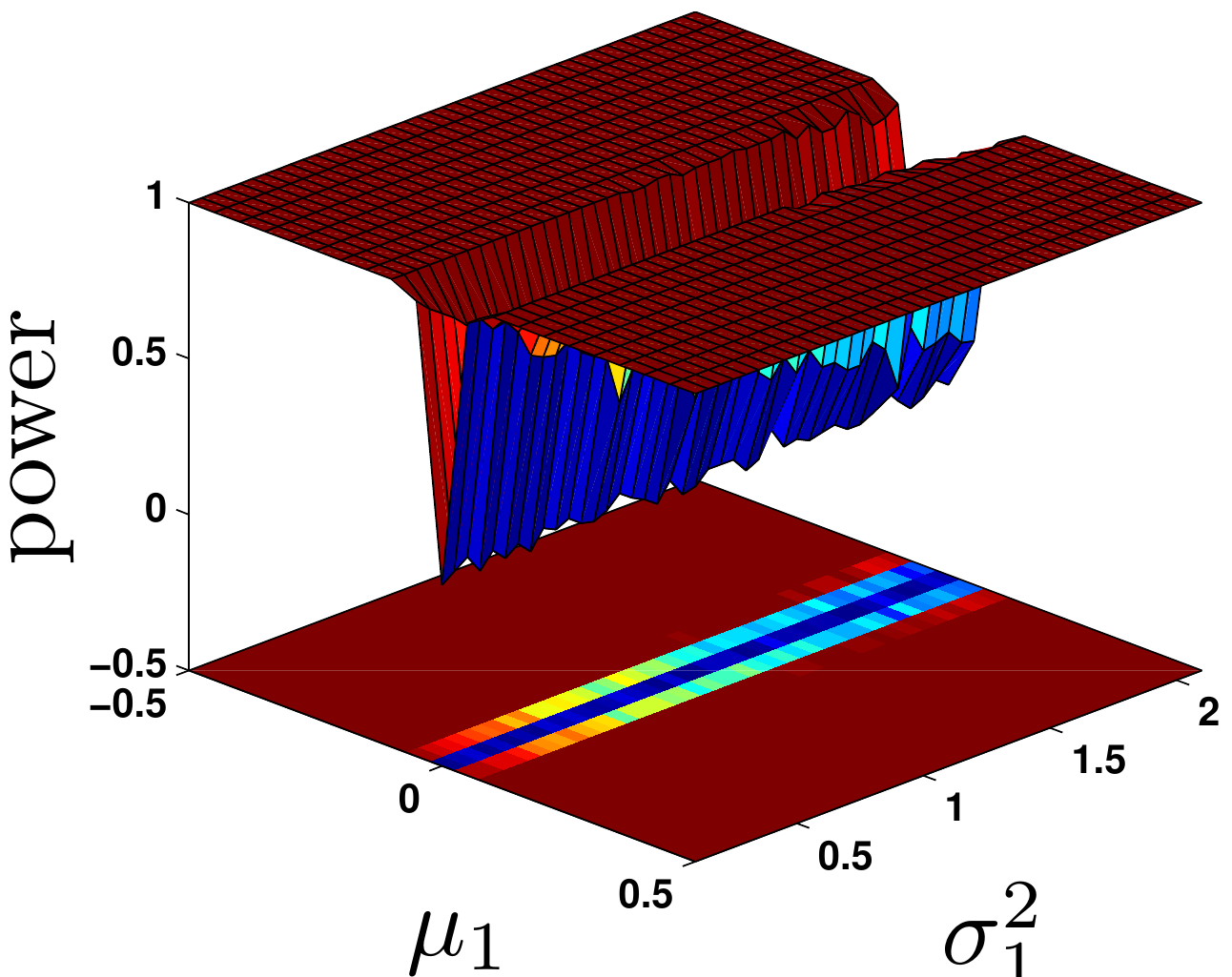}
\label{powerMDsMD}
}
\caption{Power surface of the MD-MD and the MD-scaled MD for testing equality of means for distributions $F_1=N(\mu_1,\sigma_1^2I_d)$ and $F_2=N(0,I_d)$. When sample sizes are unbalanced, the MD-scaled MD test attains the correct level while the MD-MD does not.}
\label{equalPower}
\end{figure}

We study the empirical power of the MD-MD and MD-scaled MD for testing equality of means when sample sizes are unbalanced. We set $m = 50, n=100$ and make $m$ draws from $F_1 = N(\mu_1,\sigma_1^2 I_d)$ and $n$ draws from $F_2 = N(0,I_d)$ for $d=500$. The sample sizes and dimension are chosen to reflect a HDLSS setting. The significance level is set at $\alpha = 0.05$. Power is estimated using 1000 Monte Carlo simulations and displayed using a color spectrum from cool to warm, corresponding to the range 0 to 1.

Figure \ref{equalPower}, as in the figures in Section \ref{powersufs}, displays the simulated power surface of MD-MD and MD-scaled MD. We see that when sample sizes are unbalanced and covariances unequal $(\sigma_1^2 \ne 1)$, MD-MD does not attain the correct level. Indeed MD-MD will reject increasingly often as the signal in $\sigma_1^2$ grows. On the other hand, we see from Figure \ref{powerMDsMD} that the MD-scaled MD test attains the correct level under unbalanced sample sizes. This simulated power study also suggests that the asymptotics for the MD-scaled MD test is in effect for relatively small sample sizes and a much larger dimension.

\section{Additional Implementation Options}
\label{impl}
\subsection{Direction}
\label{Di}
The following binary linear classifiers are among many possible choices for the direction vector used in Step 1 of DiProPerm and all are implemented in the DiProPerm software:
\begin{enumerate}
\item The Mean Difference method is a simple binary linear classifier, also called the centroid method \citep{Hastie2003}, where points are assigned to the class whose centroid is closest. The normal vector to the separating hyperplane is the unit vector in the direction of the line segment connecting the centroids of each class, $(\bar X-\bar Y)$. 
\item Fisher Linear Discrimination (FLD) was an early binary linear classification method, see Chapter 11 of \cite{Mardia1979} for an introduction. FLD seeks a separation that maximizes the between sum-of-squares of the two classes while minimizing the within sum-of-squares of each class. The normal vector to the separating hyperplane is the unit vector in the direction of $W^{-1} (\bar X-\bar Y)'$ where $W$ is the $d \times d$ matrix
\[
W = \sum_{i=1}^m (X_i - \bar X)(X_i - \bar X)' + \sum_{j=1}^n (Y_j - \bar Y)(Y_j - \bar Y)'
\]
\item Support Vector Machine (SVM) is a popular binary linear classification method that minimizes training error while maximizing the margin between the two classes. See \cite{Hastie2003} for a good introduction. 
\item Distance Weighted Discrimination (DWD) is a binary linear classifier similar to SVM except each data point has some weight in the final classifier \citep{Marron2007}. DWD better avoids the data piling problem exhibited by SVM in high dimensions. 
\item Maximal Data Piling (MDP) is a binary linear classifier such that the projections of the data points from each class onto its normal direction vector have two distinct values \citep{Ahn2010}. 
\end{enumerate}
Notice that we have not included any PCA directions on this list. This is because PCA is tailored to find directions that show maximal variation, which is different from our objective of finding directions that show separation between the two-samples. A more serious disadvantage to using PCA as the direction in step 1 of DiProPerm is that the univariate two-sample test statistic calculated in step 2 of DiProPerm would be invariant under relabelings. 
\subsection{Projection and univariate statistic}
\label{Pro}
In the second step of DiProPerm, we project the data onto the direction in step one and compute a univariate two-sample statistic on the projected values. Large values of the test statistic indicate departure from the null hypothesis. The following univariate two-sample statistics are among many reasonable choices for the DiProPerm procedure and all are implemented in the DiProPerm software.
\begin{enumerate}
\item Two Sample t statistic
\item Difference of sample means
\item Difference of sample means scaled, as in Equation \eqref{MDsMDeq}
\item Difference of sample medians
\item Difference of sample medians, divided by the  median absolute deviation.
\item Area Under the Curve (AUC), from Receiver Operating Characteristic (ROC) curve
\item Paired sampling t-statistic
\end{enumerate}
It is of interest to note that the classical Hotelling $T^2$ statistic is a special case of the general DiProPerm framework. The FLD direction vector and the difference of sample means combination gives the statistic $(\bar X - \bar Y) W^{-1} (\bar X - \bar Y)'$. This is in fact the Hotelling $T^2$ test statistic scaled by a factor of $\frac{1}{n-2}\frac{n}{n_1 n_2}$. To see this, recall the Hotelling $T^2$ statistic is
\[
T^2 = \frac{n_1 n_2}{n} (\bar X - \bar Y) S_u^{-1} (\bar X - \bar Y)'
\]
where 
\[
S_u = \frac{\sum_{i=1}^m (X_i - \bar X)(X_i - \bar X)' + \sum_{j=1}^n (Y_j - \bar Y)(Y_j - \bar Y)'}{n-2} = \frac{W}{n-2}
\]
The MD-FLD statistic is $\frac{1}{n-2}\frac{n}{n_1 n_2} T^2$ .

\subsection{Permutation}
\label{Perm}
In the final step of DiProPerm, an approximate permutation test is conducted to assess the significance of the test statistic in step two. Our permutation test is approximate because we perform a large number of random rearrangements of the labels on the observed data points, rather than all possible rearrangements. There are three kinds of indicators we commonly use and all are implemented in the DiProPerm software:
\begin{enumerate}
\item[1.] Empirical p-value: this is calculated as the proportion of the rearrangement test statistics that exceed the original test statistic. The empirical $p$-value has the disadvantage of often being zero. We may wish to compare two separations to see which is more significant. This motivates the next quantity.
\item[2.] Gaussian fit p-value: we fit a Gaussian distribution to the permutation test statistics and based on this calculate the percentage of rearrangement test statistics that exceed the original test statistic. (The term p-value is used loosely here). We do this not because we believe the permutation statistics are actually Gaussian, but because this provides a basis on which we can compare two DiProPerm results. In certain settings where the Gaussian fit $p$-value may suffer from round-off error, we use the next quantity as an alternative. 
\item[3.] $z$-score: we fit a Gaussian distribution to the permutation test statistics and calculate the corresponding $z$-score of the original test statistic with respect to the fitted distribution. 
\end{enumerate}
When interpreting the results of DiProPerm tests, it is generally useful to print all three indicators. When it is non-zero, the empirical $p$-value is the most interpretable. When it is zero we next look to the Gaussian fit p-value. Finally if the Gaussian fit p-value suffers from round-off error, the z-score is preferable.

\section{HDLSS calculations}
Let $X \sim N(0,\sigma_x^2 I_d)$ and $Y \sim N(0,\sigma_y^2 I_d)$. We will study the asymptotic behavior of the distance between $X$ and $Y$. We have simply by definition
\[
|| X - Y||^2 / (\sigma_x^2 + \sigma_y^2) \sim \chi^2(d).
\]
Then by the Central Limit Theorem, 
\[
\sqrt d \left( \frac{|| X - Y||^2 / (\sigma_x^2 + \sigma_y^2)}{\sqrt 2 d} - \frac{1}{\sqrt 2} \right) \to N(0,1)
\]
as $d \to \infty$. Applying the Delta Method, we get
\[
\sqrt d \left( \frac{|| X - Y ||}{2^{1/4} \sqrt{ (\sigma_x^2 + \sigma_y^2) d	}} - \frac{1}{2^{1/4}} \right) = O_P(1)
\]
and thus 
\[
|| X - Y || = \sqrt{ (\sigma_x^2+\sigma_y^2) d } + O_P(1)
\]

\end{document}